\documentclass[11pt,reqno]{article}
\usepackage[margin=1in]{geometry}
\usepackage{dsfont, amssymb,amsmath,amscd,latexsym, amsthm, amsxtra,amsfonts}
\usepackage{lineno}
\usepackage[all]{xy}
\usepackage[active]{srcltx}
\usepackage{tikz}
\usepackage{booktabs}
\usepackage[round,authoryear]{natbib}
\usepackage{bbm}
\usepackage{enumerate}
\usepackage{mathrsfs}
\usepackage{graphicx}
\usepackage{caption}
\usepackage{subcaption}
\usepackage{comment}
\usepackage{mathtools}
\usepackage{cases}
\usepackage{tcolorbox}

\tcbuselibrary{most}

\usetikzlibrary{calc,arrows}
\usepackage{verbatim}
\usepackage{color}
\usepackage{epstopdf}

\usepackage{bm}

\usepackage[title]{appendix}

\newtheorem{theorem}{Theorem}[section]
\newtheorem{assumption}[theorem]{Assumption}

\newtheorem{corollary}[theorem]{Corollary}
\newtheorem{definition}[theorem]{Definition}

\newtheorem{lemma}[theorem]{Lemma}

\newtheorem{proposition}[theorem]{Proposition}
\newtheorem{remark}[theorem]{Remark}
\numberwithin{equation}{section}

\newcommand{\md}{\mathrm{d}}
 \usepackage[pdfstartview=FitH, bookmarksnumbered=true,bookmarksopen=true, colorlinks=true, pdfborder={0 0 1}, citecolor=blue, linkcolor=blue,urlcolor=blue]{hyperref}

\hypersetup{colorlinks, citecolor=blue}
\hypersetup{
    colorlinks=true,
    linkcolor=blue,        
    anchorcolor=blue,      
    citecolor=blue,      
    filecolor=blue,        
    menucolor=blue,       
    runcolor=blue,         
    urlcolor=blue,         
}

\usepackage{graphics}
\graphicspath{{figures/}}

\newcommand{\mR}{\mathbb{R}}

\newcommand{\bmo}{\textup{BMO}}

\title{Equilibrium Portfolio Selection under Utility-Variance Analysis\\ of Log Returns in Incomplete Markets}
\author{Yue Cao\thanks{Department of Mathematical Sciences, Tsinghua University, Beijing, China. \url{ caoyue24@mails.tsinghua.edu.cn}} 
\and Zongxia Liang\thanks{Department of Mathematical Sciences, Tsinghua University, Beijing, China. \url{ liangzongxia@tsinghua.edu.cn}}
\and Sheng Wang\thanks{Department of Statistics and Actuarial Science, The University of Hong Kong, Pokfulam Road, Hong Kong. \url{sheng-wa15@tsinghua.org.cn}}
\and Xiang Yu\thanks{Department of Applied Mathematics,  The Hong Kong Polytechnic University, Kowloon, Hong Kong. \url{xiang.yu@polyu.edu.hk}}
}
\date{}

\begin{document}
\maketitle
\vspace{-0.3in}
\begin{abstract}
This paper investigates a time-inconsistent portfolio selection problem in the incomplete market model, integrating expected utility maximization with risk control. The objective functional balances the expected utility and variance on log returns, giving rise to time inconsistency and motivating the search for a time-consistent equilibrium strategy. We characterize the equilibrium via a coupled quadratic backward stochastic differential equation (BSDE) system and establish the existence results in two special cases: (i) the two Brownian motions driving the price dynamics and the factor process are independent ($\rho=0$); (ii) the trading strategy is constrained to be bounded. For the general case with correlation coefficient $\rho\neq 0$, we introduce the notion of an approximate time-consistent equilibrium. Employing the solution structure from the equilibrium in the case $\rho=0$, we can construct an approximate time-consistent equilibrium in the general case with an error of order $O(\rho^2)$. Numerical examples and financial insights are also presented based on deep learning algorithms.

\vspace{0.1in}	
 
\end{abstract}

\section{Introduction}
The optimal portfolio management problem has always been a core research topic in quantitative finance, traditionally approached via the mean-variance (MV) criterion proposed by \citet{Markowitz_1952} and the expected utility (EU) theory developed in \cite{merton_1969,MERTON1971373}. The classical MV criterion captures the trade-off between the expected return and the risk (variance). In contrary, the utility maximization framework incorporates the investor’s risk aversion level into decision making without concerning the risk level of terminal wealth. 

One natural extension is to integrate the risk management into the utility maximization problem, dictating the optimal trading strategy at a reasonable risk level. The research in this direction has gained an upsurge of attention recently. As a pioneer attempt, \cite{Bsask_2001} study the utility maximization on portfolio and consumption plans by imposing the Value-at-Risk (VaR) constraints on the terminal wealth to encode the agent's concern on the risk level.  \cite{Wong_2017} examine a utility–deviation–risk portfolio selection problem by reformulating it as an equivalent nonlinear moment system, where the objective combines expected utility with a deviation-based risk measure, and derives the optimal  strategy in a Black-Scholes model.  \cite{Bensouusan_2022} recently combine the utility maximization problem with the variance minimization of terminal wealth, and recasts the problem as a mean field-type control (MFC) problem using a coupled system of Hamilton–Jacobi–Bellman and Fokker–Planck equation, and characterizes the optimal consumption and trading strategies in the mean-field context.

Our paper adopts an integrated formulation that aims to optimize the trade-off between the expected utility and the variance of portfolio in incomplete market models. Unlike the  setting that focuses on the capital amount (see, e.g., \cite{Wong_2017} and \cite{Bensouusan_2022}), we target at the log returns of the portfolio. This choice follows a growing line of research on portfolio selection (see \cite{dai_dynamic_2021}, \cite{Peng_2023}, \cite{Guan2025}, and references therein). Specifically, the objective functional in our paper at time $t$ is defined as
\begin{align}\label{eq:J}
J(t,\pi)=\mathbb{E}_t\left[U\left(\log\left(\frac{W^{\pi}_T}{W^{\pi}_t}\right)\right)\right]-\frac{\gamma_t}{2}\textup{Var}_t\left[\log\left(\frac{W^{\pi}_T}{W^{\pi}_t}\right)\right].
\end{align}

Nevertheless, the log returns and the variance term in the objective functional render the problem time-inconsistent. That is, a strategy being optimal today may no longer remain optimal at future dates. \cite{Strotz1955} classifies agents facing time-inconsistency into three types.
The naive agent always re-optimizes the current objective, ignoring time-inconsistency and thus constantly changing strategies. The precommitted agent fixes an initial optimal plan, knowing it will later become suboptimal; the plan remains time-inconsistent because it is only optimal at the initial time. A sophisticated agent cannot precommit but anticipates future deviations, viewing her future selves as strategic counterparts and seeking a time-consistent equilibrium strategy that each self optimally follows—namely, the intrapersonal equilibrium.

Contrary to \cite{Wong_2017} and \cite{Bensouusan_2022}, who focus on the precommitted strategy, the present paper aims to characterize a time-consistent equilibrium strategy by a sophisticated agent under game theoretical thinking with future selves, whose precise definition in continuous time was first proposed in \cite{Ekeland2006noncommitmentsubgameperfectequilibrium}\footnote{Since then, a strand of literature on time-inconsistent control problems in continuous-time setting has emerged, see e.g. \cite{Basak_dynamic_2010}, \cite{hu_time-inconsistent_2012, hu_time-inconsistent_2017}, \cite{Bjork_2014}, \cite{dai_dynamic_2021} for  MV problems, \cite{Ekeland2006noncommitmentsubgameperfectequilibrium}, \cite{Ekeland_2007}, \cite{Hamaguchi_2021} for non-exponential discounting problems.}. Moreover, we consider the incomplete market model, where the financial market consists of a risk-free asset and a single risky asset\footnote{For simplicity, we focus on one risky asset, as the analysis can be extended directly to multiple risky assets.}. The wealth process is driven by a standard Brownian motion $B$, while the asset prices are influenced by an exogenous factor process driven by another Brownian motion $\bar{B}$. The correlation coefficient of these two Brownian motions is denoted by $\rho$.
Similar to \cite{Hamaguchi_2021} and \cite{dai_dynamic_2021}, we can employ the perturbation and calculus of variation method to establish a necessary condition to characterize an open-loop time-consistent equilibrium strategy in Theorem \ref{nes} for our utility maximization problem with variance minization via a two-dimensional system of coupled BSDEs \eqref{factor} (see also \eqref{simplified BSDE} for a simplified yet still coupled system), which has not been addressed before in the existing literature.

In general, the system \eqref{simplified BSDE} consists of two quadratic BSDEs coupled by the derived equilibrium strategy in 
quadratic form. Classical existence results for the multidimensional quadratic BSDEs mainly fall into two categories: the fixed-point approach based on contraction mappings, which ensures uniqueness of solutions (see \cite{HU20161066}, \cite{Luo_2017} and \cite{FAN2023105}), and the construction of uniformly convergent approximations via coefficient regularization (see \cite{xing_class_2018}, \cite{Jackson_2022} and \cite{JACKSON2023}). However, neither of these methods applies to our problem (see Appendix \ref{secC} for the detailed explanation). 

In response to the challenge of the unconventional BSDE system, we first study two special cases where the existence of solution can be obtained and the verification theorem can be exercised: (i) the case with zero correlation coefficient ($\rho=0$) for two Brownian motions $B$ and $\bar{B}$; (ii) the case with trading constraint in a bounded set.

For the general case with $\rho\neq 0$ and without trading constraint, 
we can no longer establish the existence of solution to the more complicated BSDE system. However, as a new contribution to the existing literature, we  show that by employing the solution structure from the time-consistent equilibrium strategy in the special case with $\rho=0$, one can successfully construct an approximate time-consistent equilibrium (see Definition \ref{ANE}) in the general case with $\rho\neq 0$ when the correlation coefficient $\rho$ is sufficiently small. More importantly, under some additional regularity conditions, we can rigorously show that the approximation error using the approximate time-consistent equilibrium  of order $O(\rho^2)$, see Theorem \ref{rhoNE}. This approximation result using the explicit equilibrium structure in the special case with $\rho=0$ significantly simplifies the  decision making in some general incomplete market models with small $\rho$ such that it is sufficient to implement the more tractable control process while the time-consistency can be retained in an approximate sense.

We also highlight the main differences between the present paper and some related studies in \cite{Gu_2020} and \cite{Bensoussan_2025}. \cite{Gu_2020} studies a time-consistent Hamilton–Jacobi–Bellman (HJB) equation to characterize the equilibrium strategy for a utility–deviation risk control problem. They obtain the explicit equilibrium value function and the corresponding equilibrium strategy in a homogeneous utility–deviation risk problem with CRRA utility based on a specific choice of risk-aversion function. In contrast, we focus on the optimization of the trade-off between the exponential utility and variance over 
log returns and characterize the open-loop time-consistent equilibrium via a system of coupled quadratic BSDEs, for which we establish existence results in some cases and provide an approximate equilibrium in the general case. \cite{Bensoussan_2025} study a stochastic control problem incorporating higher-order moments in the complete market model and provide the characterization of the time-consistent equilibrium strategy. On the contrary, we consider an incomplete stochastic factor model that significantly increases the mathematical complexity to establish the existence of time-consistent equilibrium. 

To illustrate our theoretical results, we further employ a deep learning–based numerical scheme, originally proposed by \cite{e_deep_2017} and \cite{Han_2018}, to solve the BSDE systems arising in the three cases—namely, the case with $\rho=0$, the case with trading constraints, and the case of the approximate equilibrium for $\rho\neq 0$. We find that when the correlation coefficient $\rho=0$, the equilibrium investment ratio exhibits a slight upward trend as time approaches the terminal horizon, differing from the observation in \cite{dai_dynamic_2021} such that the investment strategy remains constant. When $\rho<0$, however, our results are consistent with \cite{dai_dynamic_2021}, showing a downward trend in the investment ratio when time approaches the terminal horizon. In addition, we perform sensitivity analysis with respect to the risk-aversion parameter $\zeta$ in the utility function  and the weighting coefficient $\gamma$ between the expected utility and the variance. The numerical results illustrate that the equilibrium strategy exhibits a reasonable monotonic pattern with respect to both $\zeta$ and $\gamma$. Moreover, the numerical results also shows that the approximate Nash equilibrium strategy and its resulting value function are relatively close to those under the true constrained equilibrium when the trading constraint bound is large.

The remainder of the paper is organized as follows. Section \ref{sec:formulation} formulates the time inconsistent control problem as the trade-off between the utility maximization and the variance minimization over logarithmic portfolio returns, for which a time-consistent equilibrium is introduced. Section \ref{section3} provides the characterization of the time-consistent equilibrium by a quadratic BSDE system. Section \ref{sec:factor} investigates the existence of solution to the BSDE system in the stochastic factor model in two special cases: (i) when two Brownian motions in the price dynamics and in the factor model are independent; (ii) when the trading strategy is constrained in a compact set. Section \ref{sec:approximate} establishes the approximate time-consistent equilibrium in the general case when two Brownian motions are correlated with $\rho\neq 0$ using the same solution structure from time-consistent equilibrium solution in the model with $\rho=0$ with the error of order $O(\rho^2)$. Some numerical illustrations and financial implications are presented in Section \ref{sec:numerical}. Section \ref{sec:conclusion} concludes the paper with some remarks. Appendix \ref{secA} and \ref{secB} collect some inequalities and sufficient conditions for the existence theory of BSDE system that are used in previous sections, Appendix \ref{secC} elaborates some challenges of the BSDE system in the general case when $\rho\neq 0$, and Appendix \ref{sec:D} collects the proofs of some auxiliary results in previous sections.    

\ \\
\noindent
\textbf{Notations}: 
Let $T>0$ be a finite time horizon and $(\Omega,\mathcal{F},\mathbb{F},\mathbb{P})$ be a filtered complete probability space, where $\mathbb{F}=\{\mathcal{F}_t,0\leq t\leq T\}$ is the augmented natural filtration generated by two standard Brownian motions $\{B(t), 0\leq t\leq T\}$ and $\{ \bar{B}(t), 0\leq t\leq T\}$.
Their correlation coefficient $\rho$ is such that
  $\mathbb{E}[B(t)\bar{B}(t)]=\rho t$.
For $t \in [0, T]$, $p, q \geq 1$, we list below several notations that will be used frequently throughout the paper:
\begin{itemize}
\item  $H_{\textup{BMO}}$ denotes the set of $\mR$-valued $\mathbb{F}$-progressively  measurable processes such that $\|\pi\|_{\textup{BMO}}<\infty$, where
$$
\|\pi\|_{\textup{BMO}}^2:=\sup_{\tau \in \mathcal{T}_{[0,T]}} \left\| \mathbb{E}_{\tau} \left[ \int_\tau^T |\pi_s|^2 ds \right] \right\|_\infty,
$$
and $\mathcal{T}_{[0,T]}$ denotes the set of all  $\mathbb{F}$-stopping times taking values in $[0,T]$.
    \item $L^\infty_{\mathcal{F}_t}(\Omega; \mathbb{R})$ denotes the set of all $\mathbb{R}$-valued, $\mathcal{F}_t$-measurable, bounded random variables.
    
    \item $L^p_{\mathbb{F}}(\Omega; L^q(t, T; \mathbb{R}))$ denotes the set of all $\mathbb{R}$-valued,  $\mathbb{F}$-progressively  measurable processes $X$ such that
    \[
    \mathbb{E}\left[ \left( \int_t^T |X_s|^q  ds \right)^{p/q} \right] < \infty.
    \]
     When $p=q$, we write $L^p_{\mathbb{F}}(t, T; \mathbb{R}) := L^p_{\mathbb{F}}(\Omega; L^p(t, T; \mathbb{R}))$.
    \item $L^p_{\mathbb{F}}(\Omega; C([t, T]; \mathbb{R}))$ denotes the set of all $\mathbb{R}$-valued, $\mathbb{F}$-adapted, continuous processes $X$ such that
    \[
    \mathbb{E}\left[ \sup_{s \in [t,T]} |X_s|^p \right] < \infty.
    \]
    \item $\mathcal{S}^{\infty}$ denotes the set of all $\mR$-valued, $\mathbb{F}$-progressively  measurable, bounded processes.
\item A function $u$ is said to belong to $C^{2+\beta}(\mR^n)$ if and only if
\[
\sup_{x\in \mR^n} |u|+\sup_{x\in \mR^n} |u_x|+\sup_{x\in \mR^n} | u_{xx}|+\sup_{x,x'\in \mR^n}
\frac{|u_{xx}(x)- u_{xx}(x')|}
     {|x-x'|^{\beta}}
<\infty,
\]
and a function $u$ is said to belong to  $\in C^{1+\beta/2,\,2+\beta}( [0,T]\times \mR^n)$ if and only if 
$$
\begin{aligned}
&\sup_{(t,x)\in [0,T]\times \mR^n} |u|+\sup_{(t,x)\in [0,T]\times \mR^n} |u_t|+\sup_{(t,x)\in [0,T]\times \mR^n} |u_x|+\sup_{(t,x)\in [0,T]\times \mR^n} |u_{xx}|\\
&+\sup_{(t,x),(t',x')\in [0,T]\times \mR^n} 
\frac{|u_{xx}(t,x)-u_{xx}(t',x')|}
     {d\big((t,x),(t',x')\big)^{\beta}}+\sup_{(t,x),(t',x')\in [0,T]\times \mR^n}
\frac{|u_t(t,x)-u_t(t',x')|}
     {d\big((t,x),(t',x')\big)^{\beta}}
<\infty,
\end{aligned}
$$
where $d\big((t,x),(t',x')\big):= |x-x'| + |t-t'|^{1/2}$ is the parabolic distance between $(t,x)$ and $(t',x')$.
\end{itemize}

\section{Problem Formulation}\label{sec:formulation}
 We consider a financial market consisting of one risk-free bond with interest rate $r=\{r_t: t\in[0,T]\}$ and one risky asset whose price dynamics are given by
 $$
\mathrm{d}S_t=S_t\left(\mu_t\mathrm{d}t+\sigma_t \mathrm{d}B_t\right), \quad S_0=s_0\in(0,\infty),
 $$
 where the drift $\mu$, the volatility $\sigma$, and the interest rate $r$ are all $\mathbb{F}$-progressively measurable processes.  

 Throughout the paper, we impose the following  assumption on the market coefficients.
 \begin{assumption}\label{assumpmarket}
    The processes $r, \mu$ and $\sigma$  are bounded. Moreover, $\mu > r>0,\sigma>0$  for almost every $t\in[0,T]$, $\mathbb{P}$-a.s., and $\sigma^{-1}$ is also bounded.
 \end{assumption}
 Under Assumption \ref{assumpmarket}, the market price of risk $\theta:=\sigma^{-1}(\mu-r)>0$ are bounded.

 A trading strategy is a $\mathbb{F}$-progressively measurable processes $\pi$ such that $\int_0^T|\pi_t(\mu(t)-r(t))|\md t+\int_0^T|\sigma(t)\pi_t|^2\md t<\infty$ a.s., where $\pi_t$ represents the proportion of wealth invested in the risky asset at time $t$. 
 Given such a strategy, the corresponding self-financing wealth process $W^{\pi}$ evolves according to
 $$
 \md W^{\pi}_t=W^{\pi}_t\left(r_t+\pi_t(\mu_t-r_t)\md t+\pi_t\sigma_t\md B_t\right), \quad W^{\pi}_0=w_0>0.
 $$ 
 In this paper, we consider a utility-maximization framework with a variance adjustment applied to the log return of the portfolio. Specifically, for a given constant $w_0>0$\footnote{$w_0>0$ is fixed throughout the paper.} and a given trading strategy $\pi$,  the objective functional of the investor  at time $t\in[0,T)$ is defined as
 \begin{equation} \label{objectiveJ}
   J(t,\pi)=\mathbb{E}_t\left[U\left(\log\left(\frac{W^{\pi}_T}{W^{\pi}_t}\right)\right)\right]-\frac{\gamma_t}{2}\textup{Var}_t\left[\log\left(\frac{W^{\pi}_T}{W^{\pi}_t}\right)\right],  
  \end{equation}
where $\mathbb{E}_t$ and $\textup{Var}_t$ represent the conditional expectation and the conditional variance at time $t$, respectively, and $\gamma_t>0$ is the weighting parameter capturing the trade-off between expected utility and risk (variance) at time $t$.
In particular, we shall work with the CARA utility $U(x)=-\frac{1}{\zeta}e^{-\zeta x}$ with the risk aversion parameter $\zeta>0$.

To simplify the notation, let us denote  $R^{\pi}_t:=\log W^{\pi}_t$ as the log-return process, which satisfies
 \begin{align}\label{eq:R}
 \md R^{\pi}_t=\left[r_t+\pi_t(\mu_t-r_t)-\frac{1}{2}|\sigma_t \pi_t|^2 \right]\md t+\pi_t\sigma_t\md B_t,\quad R^{\pi}_0=r_0=\log w_0\in\mR.
 \end{align}
Under this notation, the objective functional \eqref{objectiveJ} can be rewritten as
$$
J(t,\pi)=\mathbb{E}_t\left[U(R^{\pi}_T-R^{\pi}_t)\right]-\frac{\gamma_t}{2}\textup{Var}_t[R^{\pi}_T-R^{\pi}_t]
   =-\frac{1}{\zeta}\mathbb{E}_t\left[\exp\left(-\zeta(R^{\pi}_T-R^{\pi}_t)\right)\right]-\frac{\gamma_t}{2}\textup{Var}_t[R^{\pi}_T].
$$

To ensure that the objective functional is well defined for the trading strategy $\pi$, we introduce the following admissibility condition.
\begin{definition}{(Admissible trading strategy)}\label{def:admis_stra}
    For any $p>1$, define
    \begin{align*}
        \Pi_p:=\left\{\pi\in H_{\textup{BMO}}:\mathbb{E}_t\left[\exp\left(-p\zeta(R^{\pi}_T-R^{\pi}_t)\right)\right]< \infty\, \, a.s. \text{ for all}\, \, t\in[0,T)\right\}.
    \end{align*}
      A trading strategy is said to be \emph{admissible} if it belongs to $\Pi_p$ for some $p>1$. And denote by   $ \Pi:=\bigcup\limits_{p>1}\Pi_p$
 the collection of all admissible trading strategies.
\end{definition}
\begin{remark}
Note that for any trading strategy $\pi$, $ R^{\pi}_T-R^{\pi}_t$ is independent of the initial value $r_0$. Thus, the set $\Pi$ does not depend on $r_0$.
   Furthermore, we require $\mathbb{E}_t\left[\exp\left(-\zeta(R^{\pi}_T-R^{\pi}_t)\right)\right]<\infty$ a.s. to ensure that the objective functional at time $t$ is well-defined. In Definition \ref{def:admis_stra}, we impose a slightly stronger integrability requirement with a power  $p>1$. As will be shown later, this strengthened condition guarantees that $\Pi$ remains stable under perturbations (see Lemma \ref{diffpro} and Corollary \ref{coro:admi:piepsilon}), which is crucial for a rigorous formulation of equilibrium strategies. Moreover, it implies that $\mathbb{E}_t\left[\left( U^{\prime}(R^{\pi}_T-R^{\pi}_t)\right)^p\right]<\infty$ a.s, a property that plays an essential role in the proof of Theorem \ref{nes}; see. e.g.,  \cite{Hamaguchi_2021} for similar integrability conditions.
    \end{remark}

The following lemma shows  $J(t,\pi)<\infty$ a.s. for any $\pi\in\Pi$.
\begin{lemma}\label{lma:2.4}
     For any $\pi\in \Pi$ and  $t\in[0,T)$, $J(t,\pi)<\infty$ a.s.. Moreover, $R^{\pi}\in L^2_{\mathbb{F}}(\Omega; C([0, T]; \mathbb{R}))$.
\end{lemma}
\begin{proof}
  See Appendix \ref{pflma:2.4}.
\end{proof}
The issue of time inconsistency arises due to the initial dependence of $R_t$ and the presence of the variance term. As a result, the classical optimal solution may no longer be meaningful because a decision that appears optimal today may be overturned at future dates. Instead, we seek an equilibrium solution using a game-theoretic approach with future selves.

To define the equilibrium strategy, we first introduce the notion of a perturbed strategy for $\pi\in\Pi$.
\begin{definition}\label{def:per}
    For any $\pi\in\Pi$, $t\in[0,T)$, $\varepsilon\in(0,T-t)$ and $\eta\in L_{\mathcal{F}_t}^{\infty}(\Omega,\mathbb{R})$, define the perturbed strategy $\pi^{t,\varepsilon,\eta}\in \Pi$ (see Corollary \ref{coro:admi:piepsilon} for the proof of admissibility) by
 $$ \pi_s^{t,\varepsilon,\eta}={\pi}_s+\eta \mathds{1}_{s\in[t,t+\varepsilon)}, \quad s\in[0,T],
 $$
 where $\mathds{1}$ denotes  the indicator function.
\end{definition}
For any $t\in [0,T)$, $\eta\in L_{\mathcal{F}_t}^{\infty}(\Omega,\mathbb{R}^d)$, $\varepsilon\in (0,T-t)$ and $\pi\in\Pi$, let  
$\xi^{t,\varepsilon,\eta,\pi}:=R^{\pi^{t.\varepsilon,\eta}}-R^{\pi}$ be the difference of log return processes between $\pi^{t,\varepsilon,\eta}$ and $\pi$. Then, $\xi^{t,\varepsilon,\pi}:=\xi^{t,\varepsilon,\eta,\pi}$\footnote{For notation simplicity, we drop the superscript $\eta$.} is the unique solution to the SDE\footnote{Formally, both $R^{\pi^{t.\varepsilon,\eta}}$ and $\xi^{t,\varepsilon,\pi}$ can be defined on the entire interval  $[0,T]$. Nevertheless, in view of the definition of equilibrium, it suffices to consider them only on $[t,T]$.} 
 \begin{equation}\label{diff}
 \left\{\begin{aligned}
     &\md \xi_s^{t,\varepsilon,\pi}=\left(a(s,\pi^{t,\varepsilon,\eta}_s)-a(s,\pi_s)\right)\md s+\sigma_s \eta 1_{s\in [t,t+\varepsilon)}\md B_s,\quad s\in[t,T],\\
     &\xi_t^{t,\varepsilon,\pi}=0,
 \end{aligned}
 \right.
 \end{equation}
 where  $a(s,\pi):=r_s+(\mu_s-r_s)\pi-\frac{1}{2}|\sigma_s\pi|^2$ . In the following lemma, we establish several fundamental properties of $\xi^{t,\varepsilon,\pi}$.
 
 \begin{lemma}\label{diffpro}
     Let $\pi\in\Pi$, $t\in[0,T)$ and $\eta\in L_{\mathcal{F}_t}^{\infty}(\Omega,\mathbb{R})$ be fixed. Then we have the
following three assertions:
     \begin{enumerate}
         \item For any $\varepsilon\in (0,T-t)$, $\xi^{t,\varepsilon,\pi}\in L_{\mathbb{F}}^{k}(\Omega;C([t,T];\mR))$ for any $k\ge1$.
         \item For any $k\ge1$ and  $\varepsilon\in(0,T-t)$, there exists a constant $C(k,\|\pi\|_{\textup{BMO}})$\footnote{ Throughout the remainder of the paper, the constant $C$ may vary from line to line. For simplicity, we will omit its dependence on market parameters such as $\|r\|_{\infty},\|\mu\|_{\infty},\|\sigma\|_{\infty},T$, etc. Instead, we will explicitly indicate the dependence of $C$ on the trading strategy $\pi$ as well as on the perturbation parameters $\varepsilon$ and $\eta$ (see Definition \ref{def:per}) when they appear.} such that
$$
\mathbb{E}_t\left[\sup\limits_{s\in[t,T]}|\xi_s^{t,\varepsilon,\pi}|^{2k}\right]\le C(k,\|\pi\|_{\bmo})(\varepsilon|\eta|^2)^k,\quad a.s..
$$
         \item For any  $c\in\mathbb{R}$, it holds that $ \sup\limits_{\varepsilon\in (0,T-t)}\mathbb{E}_t\left[\exp(c \xi^{t,\varepsilon,\pi}_T)\right]\leq C(|c|,\|\pi\|_{\textup{BMO}},\|\eta\|_{\infty}), \, \mathbb{P}-a.s.$. Moreover,  $C(|c|,\|\pi\|_{\textup{BMO}},\|\eta\|_{\infty})$  is non-decreasing with respect to $\|\eta\|_{\infty}$. 
     \end{enumerate}
 \end{lemma}
 \begin{proof}
See Appendix \ref{pflm:diffpro}.
\end{proof}
 \begin{remark}
     The above result is similar to Lemma 2.5 in \cite{Hamaguchi_2021}. However, the difference process $\xi^{t,\varepsilon,\pi}$ in \eqref{diff} contains quadratic terms of the strategy $\pi\in\Pi$, and therefore does not satisfy the properties in \cite{Hamaguchi_2021} under arbitrary trading strategies. This motivates us to impose a BMO condition in the definition of admissible strategies. As a consequence, the proof differs from that of Lemma 2.5 in \cite{Hamaguchi_2021}.
 \end{remark}
 \begin{corollary}\label{coro:admi:piepsilon}
      Assume that $\pi\in\Pi$. Then for any $t\in [0,T)$, $\epsilon\in(0,T-t)$ and $\eta\in L_{\mathcal{F}_t}^{\infty}(\Omega,\mathbb{R})$,   
     $\pi^{t,\varepsilon,\eta}\in \Pi$. 
 \end{corollary}
 \begin{proof}
See Appendix \ref{pfcoro:admi:piepsilon}.
 \end{proof}
  Next we give the definition of our (open-loop) time-consistent equilibrium strategy using the intra-personal game theoretic thinking by perturbation, which is inspired by \cite{hu_time-inconsistent_2012,hu_time-inconsistent_2017}.
 \begin{definition}
     For $\hat{\pi}\in\Pi$,  $\hat{\pi}$ is called a time-consistent equilibrium strategy if 
     $$   \limsup\limits_{\varepsilon\rightarrow 0}\frac{J(t,\hat{\pi}^{t,\varepsilon,\eta})-J(t,\hat{\pi})}{\varepsilon}\le 0 \quad a.s.,
     $$
     for any $t\in [0,T)$ and  any $\eta\in L_{\mathcal{F}_t}^{\infty}(\Omega,\mathbb{R})$. 
 \end{definition}
 \begin{remark}\label{rmk:notation}
      In the following sections, we often fix a candidate equilibrium strategy $\hat{\pi}\in\Pi$ or an arbitrary strategy $\pi\in\Pi$ for analysis.  In proving the necessary conditions or verifying theorems for equilibrium strategies, we will work with the perturbed log return processes ${R}^{\hat{\pi}^{t,\epsilon,\eta}}$ or $ {R}^{\pi^{t,\epsilon,\eta}}$ under a fixed perturbation $\eta$. To ease presentation,  we denote the perturbed strategies $\hat{\pi}^{t,\varepsilon,\eta}$ and $\pi^{t,\varepsilon,\eta}$ simply by $\hat{\pi}^{t,\varepsilon}$ and $\pi^{t,\varepsilon}$, respectively. Likewise, we write $\hat{R}$ and $R$ for the log-return processes $R^{\hat{\pi}}$ and $R^{\pi}$, and $\hat{R}^{t,\varepsilon}$ and $R^{t,\varepsilon}$ for the perturbed log-return processes $R^{\hat{\pi}^{t,\varepsilon}}$ and $R^{\pi^{t,\varepsilon}}$.
The difference processes $\xi^{t,\varepsilon,\eta,\hat{\pi}}$ and $\xi^{t,\varepsilon,\eta,\pi}$ and are denoted simply by $\hat{\xi}^{t,\varepsilon}$ and $\xi^{t,\varepsilon}$.    For simplicity, we also use the notations $\hat{u}:=\sigma\hat{\pi}$ and $u:=\sigma\pi$ in the following sections.
\end{remark}

\section{A Necessary Condition of Time-Consistent Equilibria}\label{section3}
In this section, by establishing several properties of admissible strategies, we derive a necessary condition for a time-consistent equilibrium characterized by a coupled quadratic BSDE system. The main theorem of this section is stated below.
 \begin{theorem}\label{nes}
     If $\hat{\pi}\in\Pi_p$ ($p>1$) is a time-consistent equilibrium strategy, there exist $(Y,\tilde{Y})$ and $(Z^1, Z, \tilde{Z}^1, \tilde{Z})$ such that
     \begin{enumerate}
         \item $Y$ and $\tilde{Y}$ are continuous and adapted processes such that $\exp\left(-\zeta(\hat{R}+Y)\right)\in L^p_{\mathbb{F}}(\Omega;C([0,T];\mathbb{R}))$ and $\tilde{Y}\in L^2_{\mathbb{F}}(\Omega; C([0, T]; \mathbb{R}))$.  Moreover, $(Z^1,\tilde{Z}^1,Z,\tilde{Z})$ are $\mathbb{R}$-valued progressively measurable processes such that
         $$\int_0^T \left(|Z_s^1|^2+|Z_s|^2+|\tilde{Z}_s^1|^2+|\tilde{Z}_s|^2\right)\md s<\infty,\quad \mathbb{P}-a.s..$$
     
    \item $(Y,\tilde{Y})$ and $(Z^1, Z, \tilde{Z}^1, \tilde{Z})$ satisfy the following quadratic BSDE system:
     \begin{equation}\label{factor}
     \left\{
     \begin{aligned}
         &\md Y_s=\left(\frac{\zeta}{2}|Z_s^1+\sigma_s\hat{\pi}_s+\rho Z_s|^2+\frac{\zeta(1-\rho^2)}{2}|Z_s|^2-a(s, \hat{\pi}_s)\right)\md s+Z_s^1\md B_s+Z_s\md \bar{B}_s ,\\
         &\md \tilde{Y}_s=-a(s,\hat{\pi}_s)\md s+\tilde{Z}^1_s\md B_s+\tilde{Z}_s\md \bar{B}_s,\\
         &Y_T=0,\\
         &\tilde{Y_T}=0.
     \end{aligned}
     \right.
     \end{equation}
     \end{enumerate}
    Furthermore, the equilibrium strategy is characterized by
     \begin{align}\label{eq:barpi}
     \hat{\pi}=\sigma^{-1} \frac{ e^{-\zeta Y}\theta-\zeta e^{-\zeta Y}(Z^1+\rho Z)-\gamma(\tilde{Z}^1+\rho\tilde{Z})}{(\zeta+1)e^{-\zeta Y}+\gamma}.
     \end{align}
 \end{theorem}
\begin{remark}\label{complete}
When the market model becomes complete, i.e., the second Brownian motion $\bar{B}$ is absent, the BSDE system \eqref{factor} degenerates to
\begin{equation}\label{completeequ}
     \left\{
     \begin{aligned}
         &\md Y_s=\left(\frac{\zeta}{2}|Z_s+\hat{u}_s|^2-a(s,\hat{\pi}_s)\right)\md s+Z_s\md B_s ,\\
         &\md \tilde{Y}_s=-a(s,\hat{\pi}_s)\md s+\tilde{Z}\md B_s,\\
         &Y_T=0,\\
         &\tilde{Y_T}=0.
     \end{aligned}
     \right.
     \end{equation}
  The equilibrium strategy takes the form $\hat{\pi}=\sigma^{-1}\frac{ e^{-\zeta Y}\theta-\zeta e^{-\zeta Y}Z-\gamma\tilde{Z}}{(\zeta+1)e^{-\zeta Y}+\gamma}$.
    We observe that the BSDE system admits a solution when $Z=0$ and $\tilde{Z}=0$, thus an equilibrium is given by
    $\hat{\pi}=\sigma^{-1}\frac{e^{-\zeta Y}\theta}{(\zeta+1)e^{-\zeta Y}+\gamma}$,
   where $Y$ is the solution of the ODE with terminal condition:
    \begin{align}\label{eq:Y:complete}
    \left\{
    \begin{aligned}
        &\md Y_s=\left(\frac{\zeta}{2}|\hat{u}_s|^2-a(s,\hat{\pi}_s)\right)\md s,\\
        &Y_T=0.
    \end{aligned}
    \right.
    \end{align}
    Let $A=e^{-\zeta Y}$, then $A$ satisfies the  following ODE 
    \begin{align}\label{eq:ode:A}
       \left\{
    \begin{aligned}
    &A_t=A(f(A)),\quad t\in[0,T),\\
    &A(T)=1,
    \end{aligned}\right.
    \end{align} 
    where
    \begin{align*}
       f(A):= -\frac{\zeta(\zeta+1)A^2\theta^2}{2((\zeta+1)A+\gamma)^2}+\frac{\zeta A\theta^2}{(\zeta+1)A+\gamma}+r\zeta.
    \end{align*}
     It is clear that $A$ has a lower bound $0$.  Combining $A'=A(f(A))$ and the fact that $f(A)$ is bounded, we deduce that $A$ is bounded from above and thus there exists a unique solution and $A$ is  positive.
\end{remark}
\begin{remark}
    \cite{Hamaguchi_2021} derives a one-dimensional FBSDE as a necessary condition for an equilibrium in a general non-exponential discounting time-inconsistent problem. In our setting, however, the introduced variance must be represented by another BSDE for $\tilde{Y}$, which leads to strong coupling in our system. Although \cite{dai_dynamic_2021} study an MV framework that includes a variance term in the objective function, it is straightforward to see that the mean term can be represented by the BSDE for $\tilde{Y}$. In fact, $\tilde{Y} := \mathbb{E}_t[R_T - R_t]$ (see Lemma \ref{lma:3.4} below) corresponds precisely to the mean term at time $t$. Hence, their analysis can also focus on a one-dimensional BSDE. Unlike \cite{Hamaguchi_2021} and \cite{dai_dynamic_2021}, our problem naturally leads to a two-dimensional fully coupled quadratic BSDE system.
\end{remark}
To prove Theorem \ref{nes}, we first develop several auxiliary results. Lemma \ref{lma:3.4} characterizes the difference between the objective functional of a strategy $\pi\in\Pi_p$ and that of its perturbed counterpart. This lemma is inspired by \cite{horst2014forward}, \cite{Hamaguchi_2021}, and \cite{dai_dynamic_2021}, and constitutes a key ingredient in establishing Theorem \ref{nes}.
The main idea in proving Lemma \ref{lma:3.4} is to apply a Taylor expansion and then compute the resulting expectations via a suitably constructed BSDE with an appropriate terminal condition. The proof of Lemma \ref{lma:3.4} is similar to that of Lemma 3.2 in \cite{Hamaguchi_2021}. We include it in Appendix \ref{pflma:3.4} because our setting does not fit exactly within the framework of \cite{Hamaguchi_2021} or other existing results.
\begin{lemma}\label{lma:3.4}
    For any $p>1$ and $\pi\in\Pi_p$, there exist processes $(Y,\tilde{Y})$ and $(Z^1,Z,\tilde{Z}^1, \tilde{Z})$ satisfying the BSDEs \eqref{factor} (with $\hat{\pi}$ replaced by $\pi$)  such that, for any $t\in [0,T)$, $\eta\in L_{\mathcal{F}_t}^{\infty}(\Omega,\mathbb{R}^d)$ and $\varepsilon\in(0,T-t)$, we have, almost surely,
     
    \begin{align}\label{eq:lma:3.4}
         &J(t,\pi^{t,\varepsilon})-J(t,\pi)\nonumber\\
         =&\mathbb{E}_t\left[\int_t^{t+\varepsilon}e^{-\zeta(R_s+Y_s-R_t)}(a(s,\pi^{t,\varepsilon}_s)-a(s,\pi_s)-\zeta(Z_s^1+\sigma_s\pi_s+\rho Z_s)\cdot \sigma_s\eta )\md s\right]\nonumber\\
         &-\zeta\mathbb{E}_t\left[\int_0^1 e^{-\zeta(R_T+\lambda\xi_T^{t,\varepsilon}-R_t)}(1-\lambda)d\lambda|\xi_T^{t,\varepsilon}|^2\right]\nonumber\\&+\frac{\gamma}{2}\mathbb{E}_t\left[\int_t^{t+\varepsilon}((\sigma_s\pi_s+\tilde{Z}_s^1+\rho \tilde{Z}_s)^2-(\sigma_s\pi_s^{t,\varepsilon}+\tilde{Z}_s^1+\rho \tilde{Z}_s)^2)\md s\right].
    \end{align}
    
\end{lemma}

Lemma \ref{lma:3.4} shows that, for a fixed strategy $\pi\in\Pi_p$ ($p>1$), the equality \eqref{eq:lma:3.4} provides an explicit characterization of the difference between the objective functional of $\pi$ and that of its perturbed counterpart through the solution of the BSDE system \eqref{factor} (with $\hat{\pi}$ replaced by $\pi$). In contrast, the following Corollary \ref{corol:verifyequilibrium} provides the same characterization by using only a simplified form of \eqref{eq:lma:3.4}, under the additional assumption that the simplified BSDE system \eqref{eq:simplifiedBSDE} admits a solution. This corollary will be used in several subsequent verification theorems. The proof of Corollary \ref{corol:verifyequilibrium} follows similar lines to that of Lemma \ref{lma:3.4}, and a sketch of the proof is provided in Appendix \ref{pfcoro:verifyequilibrium}.
\begin{corollary}\label{corol:verifyequilibrium}
    For any $p>1$ and $\pi\in\Pi_p$, assume that there exist continuous and adapted processes $(Y,\tilde{Y})$ such that $\exp\left(-\zeta(R+Y)\right)\in L^p_{\mathbb{F}}(\Omega;C([0,T];\mathbb{R}))$ and $\tilde{Y}\in L^2_{\mathbb{F}}(\Omega; C([0, T]; \mathbb{R}))$, as well as  progressively measurable processes $(Z, \tilde{Z})$ such that $\int_0^T \left(|Z_s|^2+|\tilde{Z}_s|^2\right)\md s<\infty \,\, a.s..$
    Moreover, suppose that $(Y,\tilde{Y},Z,\tilde{Z})$ satisfy the following  BSDE
      \begin{equation}\label{eq:simplifiedBSDE}
     \left\{
     \begin{aligned}
         &\md Y_s=\left(\frac{\zeta}{2}|\sigma_s\pi_s+\rho Z_s|^2+\frac{\zeta(1-\rho^2)}{2}|Z_s|^2-a(s, \pi_s)\right)\md s+Z_s\md \bar{B}_s ,\\
         &\md \tilde{Y}_s=-a(s,\pi_s)\md s+\tilde{Z}_s\md \bar{B}_s,\\
         &Y_T=0,\\
         &\tilde{Y_T}=0.
     \end{aligned}
     \right.
     \end{equation}
     Then for any $t\in [0,T)$, $\eta\in L_{\mathcal{F}_t}^{\infty}(\Omega,\mathbb{R}^d)$ and $\varepsilon\in(0,T-t)$, we have, almost surely,
     $$
    \begin{aligned}
         &J(t,\pi^{t,\varepsilon})-J(t,\pi)\\
         =&\mathbb{E}_t\left[\int_t^{t+\varepsilon}e^{-\zeta(R_s+Y_s-R_t)}(a(s,\pi^{t,\varepsilon}_s)-a(s,\pi_s)-\zeta(\sigma_s\pi_s+\rho Z_s)\cdot \sigma_s\eta )\md s\right]\\
         &-\zeta\mathbb{E}_t\left[\int_0^1 e^{-\zeta(R_T+\lambda\xi_T^{t,\varepsilon}-R_t)}(1-\lambda)d\lambda|\xi_T^{t,\varepsilon}|^2\right]\\&+\frac{\gamma}{2}\mathbb{E}_t\left[\int_t^{t+\varepsilon}((\sigma_s\pi_s+\rho \tilde{Z}_s)^2-(\sigma_s\pi_s^{t,\varepsilon}+\rho \tilde{Z}_s)^2)\md s\right].
    \end{aligned}
    $$
\end{corollary}

To calculate $\lim\limits_{\varepsilon\rightarrow0} \frac{1}{\varepsilon}(J(t,\pi^{t,\varepsilon})-J(t,\pi))$, we rely on Lemmas \ref{lma:3.5} and \ref{lma:3.6} below.
Lemma \ref{lma:3.5} provides an upper bound for the second-order term appearing in the expansion of $J(t,\pi^{t,\varepsilon})-J(t,\pi)$ given in Lemma \ref{lma:3.4}, while Lemma \ref{lma:3.6} establishes the convergence of the remaining two terms after division by $\varepsilon$. The proofs of Lemmas \ref{lma:3.5} and \ref{lma:3.6} are given in Appendix \ref{pflma:3.5} and Appendix \ref{pflma:3.6}, respectively.
\begin{lemma}\label{lma:3.5}
For any $p>1$, $\pi\in\Pi_p$, $t\in [0,T)$ and $\eta\in L_{\mathcal{F}_t}^{\infty}(\Omega,\mathbb{R})$, it holds that
  $$
   \limsup\limits_{\varepsilon\rightarrow 0}\frac{1}{\varepsilon} \mathbb{E}_t\left[\int_0^1 e^{-\zeta(R_T+\lambda\xi_T^{t,\varepsilon}-R_t)}(1-\lambda)d\lambda|\xi_T^{t,\varepsilon}|^2\right]\le C(\|\eta\|_{\infty})|\eta|^2, \quad a.s.,
    $$
where $C$ is a constant independent of $\varepsilon$ and is non-decreasing with respect to $\|\eta\|_{\infty}$. 
\end{lemma}

\begin{lemma}\label{lma:3.6}
For any $p>1$ and $\pi\in \Pi_p$, consider the pair $(Y,\tilde{Y})$ and $(Z,Z^1,\tilde{Z},\tilde{Z}^1)$ in Lemma \ref{lma:3.4}.  We deduce the existence of a measurable set $E_1 \subset [0,T)$ with $\operatorname{Leb}([0,T] \setminus E_1) = 0$  such that,
for any $t \in E_1$, there exists a sequence $\{\varepsilon_n^t\}_{n \in \mathbb{N}} \subset (0,T-t)$ satisfying $\lim\limits_{n \to \infty} \varepsilon_n^t = 0$
and for any perturbation $\eta\in L_{\mathcal{F}_t}^{\infty}(\Omega;\mR)$, it holds that 
\begin{align}\label{lma:3.5.1}
&\lim\limits_{n\rightarrow\infty} \frac{1}{\varepsilon_n^t}\mathbb{E}_t\left[\int_t^{t+\varepsilon_n^t} e^{-\zeta(R_s+Y_s-R_t)}\left(a(s,\pi^{t,\varepsilon_n^t}_s)-a(s,\pi_s)-\zeta(Z_s^1+\sigma_s\pi_s+\rho Z_s)\cdot \sigma_s\eta \right)\md s\right]\nonumber\\
=&e^{-\zeta Y_t}\left((\mu_t-r_t)\eta-\frac{1}{2}\sigma^2_t\eta^2-\sigma^2_t\eta\pi_t-\zeta(Z_t^1+\sigma_t\pi_t+\rho Z_t)\cdot \sigma_t\eta\right),\quad a.s.
\end{align}
and 
\begin{align}\label{lma:3.5.2}
    &\lim\limits_{n\rightarrow\infty} \frac{1}{\varepsilon_n^t}\mathbb{E}_t\left[\int_t^{t+\varepsilon_n^t} \left((\sigma_s\pi_s+\tilde{Z}^1_s+\rho\tilde{Z}_s)^2-(\sigma_s\pi_s^{t,\varepsilon}+\tilde{Z}^1_s+\rho\tilde{Z}_s)^2\right)\md s\right]\nonumber\\
=&-2\sigma_t\eta(\sigma_t\pi_t+\tilde{Z}^1_t+\rho\tilde{Z}_t-\sigma_t^2\eta^2),\quad a.s..
\end{align}
\end{lemma}
\begin{remark}
Note that the constant $C$ in Lemma \ref{lma:3.5} depends on $\|\eta\|_{\infty}$. Fortunately, $C$ is non-decreasing in $\|\eta\|_{\infty}$, which is sufficient for establishing the necessary condition; see the proof of Theorem \ref{nes}.
Moreover, since the integrands in \eqref{lma:3.5.1} and \eqref{lma:3.5.2} involve $a(s,\pi^{t,\varepsilon_n^t}_s)$ and $(\sigma_s\pi_s^{t,\varepsilon}+\tilde{Z}^1_s+\rho\tilde{Z}_s)^2$, which depend on the perturbed strategy $\pi^{t,\varepsilon_n^t}$ over the interval $[t, t+\varepsilon_n^t]$ and are thus influenced by the perturbation $\eta$, Lemma 3.3 in \cite{Hamaguchi_2021} cannot be applied directly. 
However, because the effect of $\eta$ can be separated, we may still use a diagonal selection argument to complete the proof, with $\varepsilon_n^t$ chosen independently of $\eta$.
\end{remark}
Now we are ready to give the proof of Theorem \ref{nes}.
\begin{proof}[Proof of Theorem \ref{nes}]
    Assume that $\hat{\pi}\in\Pi$ is an  equilibrium strategy.  The existence of $(Y,\tilde{Y})$ and $(Z^1, Z, \tilde{Z}^1, \tilde{Z})$, as well as the fact that they satisfy the BSDE \eqref{factor}, follows directly from Lemma \ref{lma:3.4}. Thus, it remains to prove \eqref{eq:barpi}.
    Fix an arbitrary $\delta>0$ and take an arbitrary $t$ from the set $E_1$ in Lemma  \ref{lma:3.6} . Let $\{\delta_m\}_{m\in\mathbb{N}}$ be a sequence such that $0< \delta_m\le \delta,\; m\in\mathbb{N},$ and $\lim\limits_{m\rightarrow\infty}\delta_m=0$. For each $m\in\mathbb{N}$, define $\eta_m$ by
     $$
    \eta_m=
    \begin{cases}
         \delta_m & e^{-\zeta Y_t}\theta- e^{-\zeta Y_t}\sigma_t\hat{\pi}_t-\zeta e^{-\zeta Y_t}(\sigma_t\hat{\pi}_t+Z_t^1+\rho Z_t)-\gamma(\tilde{Z}^1_t+\sigma_t\hat{\pi}_t+\rho\tilde{Z}_t)\ge 0,\\
         -\delta_m & e^{-\zeta Y_t}\theta- e^{-\zeta Y_t}\sigma_t\hat{\pi}_t-\zeta e^{-\zeta Y_t}(\sigma_t\hat{\pi}_t+Z_t^1+\rho Z_t)-\gamma(\tilde{Z}^1_t+\sigma_t\hat{\pi}_t+\rho\tilde{Z}_t)<0.
    \end{cases}
    $$
    
    By the definition of the equilibrium, Lemma \ref{lma:3.6}, Lemma \ref{lma:3.4} as well as Lemma \ref{lma:3.5}, we obtain that there exists a sequence $\{\varepsilon_n^t\}_{n \in \mathbb{N}} \subset (0,T-t)$ satisfying $\lim\limits_{n \to \infty} \varepsilon_n^t = 0$
such that
    $$
    \begin{aligned}
            0\geq &\limsup\limits_{n\rightarrow \infty}\frac{1}{\varepsilon_n^t}\left(J(t,\hat{\pi}^{t,\varepsilon_n^t,\eta_m})-J(t,\hat{\pi})\right)\\
            =& e^{-\zeta Y_t}\left((\mu_t-r_t)\eta_m-\frac{1}{2}\sigma_t^2\eta_m^2-\sigma_t^2\hat{\pi}_t\eta_m-\zeta(Z_t^1+\sigma_t\hat{\pi}_t+\rho Z_t)\cdot \sigma_t\eta_m \right)\\&\quad\quad-\frac{\gamma}{2}\sigma_t\eta_m\left(2\sigma_t\hat{\pi}+2\tilde{Z}^1_t+2\rho\tilde{Z}_t+
            \sigma_t\eta_m\right)-C(|\delta|)(|\eta_m|^2),\quad a.s..
    \end{aligned}
    $$
  Dividing both sides of the above inequality by $\delta_m$, sending $m\rightarrow\infty$ and noting that the coefficient of $\sigma_t\eta_m$ tends to zero, we deduce that 
    $$
     e^{-\zeta Y_t}\theta- e^{-\zeta Y_t}\sigma_t\hat{\pi}_t-\zeta e^{-\zeta Y_t}(\sigma_t\hat{\pi}_t+Z_t^1+\rho Z_t)-\gamma(\tilde{Z}^1_t+\sigma_t\hat{\pi}_t+\rho\tilde{Z}_t)=0.
    $$
  Thus, \eqref{eq:barpi} holds.
\end{proof}
\section{Study of BSDE System in the Incomplete Factor Model}\label{sec:factor}
In this section, we focus on the factor model and on the two-dimensional quadratic BSDE system in \eqref{factor}, which is coupled through \eqref{eq:barpi}. The factor model provides an incomplete market setting with stochastic parameters in which the risk-free rate, the stock return rate, and the volatility rate can be expressed by a deterministic function of time $t$ and the factor $X_t$. We assume that the stochastic factor $X$ is governed by
\begin{align}\label{eq:sde:factor}
\md X_t=m(t,X_t)\md t+ \nu(t,X_t)\md \bar{B_t}, \quad X_0=x_0.
\end{align}
It is assumed that $r_t=r(t,X_t), \mu_t=\mu(t,X_t), \sigma_t=\sigma(t,X_t)$ for some functions $r,\mu,\sigma:[0,T]\times\mR\to \mR$ respectively. Moreover, we impose the following assumptions:
 \begin{assumption}\label{assump1}
$\mu(t,x)> r(t,x)>0,\sigma(t,x)>0$ and $\sigma^{-1}(t,x)$ are bounded measurable deterministic functions with respect to $(t,x)$.  
\end{assumption}

  \begin{assumption}\label{assump2}
    The coefficients of the factor model satisfy that

    \textbf{1}. The drift vector $m$ is uniformly bounded.

    \textbf{2}. There exists a constant $\lambda>0$ such that $ \lambda|z|^2\ge |z\nu(t,x)|^2\ge\frac{1}{\lambda}|z|^2$ for any $(t,x,z)\in [0,T]\times \mathbb{R}\times \mR$.

    \textbf{3}. There exists a constant $L$ such that
    $$
    |m(t,x)-m(t,x^{\prime})|+|\nu(t,x)-\nu(t,x^{\prime})|\le L|x-x^{\prime}|.
    $$
\end{assumption}
Under Assumption \ref{assump2},  SDE \eqref{eq:sde:factor} admits a unique strong solution. 
 It is noted that the log-return process $\hat{R}$ does not appear in the BSDE system \eqref{factor} (coupled through \eqref{eq:barpi}), which motivates us to consider  the solution adapted to the filtration $\mathbb{F}_{\bar{B}}$. We can therefore conjecture that $Z^1=\tilde{Z}^1=0$ and aim to find the solution to the following Markovian BSDE system:
\begin{equation}\label{simplified BSDE}
  \left\{
\begin{aligned}
    &\md X_s=m(s,X_s)\md t+\nu(s,X_s)\md \bar{B}_s, \\
    &\md Y_s=\left(\frac{\zeta}{2}|\rho Z_s+\hat{u}_s|^2+\frac{\zeta(1-\rho^2)}{2}|Z_s|^2-a(s,\hat{\pi}_s)\right)\md s+Z_s\md \bar{B}_s,\\
    &\md \tilde{Y}_s=-a(s,\hat{\pi}_s)\md s+\tilde{Z}_s\md\bar{B}_s,\\
    &X_0=x_0,\quad Y_T=0,\quad \tilde{Y}_T=0,\\
    &\hat{u}=\frac{ e^{-\zeta Y}\theta-\zeta e^{-\zeta Y}\rho Z-\gamma\rho\tilde{Z}}{(\zeta+1)e^{-\zeta Y}+\gamma},\hat{\pi}=\sigma^{-1}\hat{u}, a(\cdot,\hat{\pi})=r+(\mu-r)\hat{\pi}-\frac{1}{2}|\sigma\hat{\pi}|^2.
\end{aligned}
\right.
\end{equation}

Although simplified, \eqref{simplified BSDE} remains a two-dimensional quadratic BSDE system coupled through the derived equilibrium strategy $\hat{\pi}$, which prevents us from solving it in full generality.  Several challenges associated with solving \eqref{simplified BSDE} in the general case are discussed in Appendix \ref{secC}.

In the next two subsections, we establish the existence of an equilibrium in two special cases:
\begin{enumerate}[(i)]
    \item When $\rho = 0$, the equilibrium strategy no longer depends on $(Z,\tilde{Z})$, and the system \eqref{simplified BSDE} decouples into two separate BSDEs, for which the existence of solutions follows directly.
    \item Under a trading constraint, we introduce a modified definition of the equilibrium strategy. In this setting, we still obtain a BSDE system with the same form of \eqref{factor}, but the coupling strategy involves a projection operation (see \eqref{eq:constraint:pi}).
\end{enumerate}

\subsection{Existence of time-consistent equilibrium when $\rho=0$}
In this subsection, we consider the case $\rho=0$. Then it follows that $\hat{\pi}$ is independent of $Z$ and $\tilde{Z}$,
and thus the BSDE system \eqref{simplified BSDE}  can be decoupled as
\begin{equation}\label{rho0}
\left\{
\begin{aligned}
    &\md X_s=m(s,X_s)\md t+\nu(s,X_s)\md \bar{B}_s, \\
    &\md Y_s=\left(\frac{\zeta+1}{2}|\hat{u}_s|^2+\frac{\zeta}{2}|Z_s|^2-r(s,X_s)-\theta(s,X_s)\hat{u}_s\right)\md s+Z_s\md \bar{B}_s,\\
    &X_0=x_0,\quad Y_T=0,\\
    &\hat{u}=\frac{ e^{-\zeta Y}\theta}{(\zeta+1)e^{-\zeta Y}+\gamma},\quad \hat{\pi}=\sigma^{-1}\hat{u},
\end{aligned}
\right.
\end{equation}
and 
\begin{equation}\label{rho0BSDE2}
\left\{
\begin{aligned}
    &\md \tilde{Y}_s=-r(s,X_s)-\theta(s,X_s)\hat{u}_s+\frac{1}{2}|\hat{u}_s|^2\md s+\tilde{Z}_s\md\bar{B}_s,\\
    & \tilde{Y}_T=0 .
\end{aligned}
\right.
\end{equation}
It is straightforward to see that one can first solve BSDE \eqref{rho0} and then substitute this solution into BSDE \eqref{rho0BSDE2} to obtain the solution for the entire system.
Regarding the existence of the solution to BSDE~\eqref{rho0}, we verify Conditions (2A3), (3A2), (2A1), and (4A2) and apply Theorem 4.3 in \cite{FAN20161511}, which leads to the following existence and uniqueness result.
\begin{proposition}\label{thm:rho0BSDE}
    Under  Assumptions \ref{assump1} and \ref{assump2}, the BSDE \eqref{rho0} has a unique solution $(Y,Z)\in \mathcal{S}^{\infty}\times L^2_{\mathbb{F}}(0, T; \mathbb{R})$. 
\end{proposition}
By Proposition \ref{thm:rho0BSDE}, the candidate equilibrium strategy defined in \eqref{rho0} satisfies
$
\hat{\pi}\in \mathcal{S}^{\infty}
$. Therefore,  BSDE \eqref{rho0BSDE2} admits a unique solution $(\tilde{Y},\tilde{Z})\in \mathcal{S}^{\infty}\times L^2_{\mathbb{F}}(0, T; \mathbb{R})$ and $\tilde{Y}$ is given by
$$
\tilde{Y}_t=\mathbb{E}_t\left[\int_t^T \left(r(s,X_s)+\theta(s,X_s)\hat{u}_s-\frac{1}{2}|\hat{u}_s|^2\right)\md s\right].
$$
In the following, we verify that $\hat{\pi}$ is an equilibrium strategy by means of BSDEs~\eqref{rho0BSDE2} and~\eqref{thm:rho0BSDE}, along with Lemma~\ref{lma:3.4}.
This result is stated in the theorem below.
\begin{theorem}\label{verification0}
  Under Assumptions~\ref{assump1} and~\ref{assump2},  when $\rho=0$,  the  strategy $\hat{\pi}$ defined in~\eqref{rho0} is  an equilibrium strategy.
\end{theorem}
\begin{proof}
For notational simplicity, we adopt the same notation as in Remark \ref{rmk:notation} for  any fixed $t\in [0,T)$, $\eta\in L_{\mathcal{F}_t}^{\infty}(\Omega,\mathbb{R}^d)$ and $\varepsilon\in(0,T-t)$.

We first verify that $\hat{\pi}\in \Pi$. Because $\hat{\pi}$ is bounded, it follows that $\hat{\pi}\in H_{\textup{BMO}}$. 
It is straightforward to show that
$$
\mathbb{E}\left[\exp\left(c\left|\int_t^T\sigma_s\hat{\pi}_s\md B_s\right|\right)\right]\le 2\exp\left(\frac{c^2}{2}\|\sigma\hat{\pi}\|_{\infty}^2(T-t)\right), \quad \forall c>0.
$$
Combining this with $\hat{R}_T-\hat{R}_t=\int_t^T\left(r_s+(\mu_s-r_s)\hat{\pi}_s-\frac{1}{2}|\sigma_s\hat{\pi}_s|^2\right)\md s+\sigma_s\hat{\pi}_s\md B_s$, 
we conclude that   $\mathbb{E}_t\left[\exp\left(-p\zeta(\hat{R}_T-\hat{R}_t)\right)\right]< \infty,\, a.s.$ for any $p>1$. Thus, $\hat{\pi}\in\Pi$.

 Next, we show that $\hat{\pi}$ is indeed an equilibrium strategy. By Corollary \ref{corol:verifyequilibrium}, we obtain
$$
    \begin{aligned}
         &J(t,\hat{\pi}^{t,\varepsilon})-J(t,\hat{\pi})\\
         =&\mathbb{E}_t\left[\int_t^{t+\varepsilon}e^{-\zeta(\hat{R}_s+Y_s-\hat{R}_t)}(a(s,\hat{\pi}^{t,\varepsilon}_s)-a(s,\hat{\pi}_s)-\zeta\sigma_s\hat{\pi}_s\cdot \sigma_s\eta )\md s\right]\\
         &-\zeta\mathbb{E}_t\left[\int_0^1 e^{-\zeta(\hat{R}_T+\lambda\hat{\xi}_T^{t,\varepsilon}-\hat{R}_t)}(1-\lambda)d\lambda|\hat{\xi}_T^{t,\varepsilon}|^2\right]\\&+\frac{\gamma}{2}\mathbb{E}_t\left[\int_t^{t+\varepsilon}(|\sigma_s\hat{\pi}_s|^2-|\sigma_s\hat{\pi}_s^{t,\varepsilon}|^2)\md s\right]\\
 \leq &\mathbb{E}_t\left[\int_t^{t+\varepsilon}e^{-\zeta(\hat{R}_s+Y_s-\hat{R}_t)}(a(s,\hat{\pi}^{t,\varepsilon}_s)-a(s,\hat{\pi}_s)-\zeta\sigma_s\hat{\pi}_s\cdot \sigma_s\eta )\md s\right]\\
 &+\frac{\gamma}{2}\mathbb{E}_t\left[\int_t^{t+\varepsilon}(|\sigma_s\hat{\pi}_s|^2-|\sigma_s\hat{\pi}_s^{t,\varepsilon}|^2)\md s\right]\\
         =& \mathbb{E}_t\left[\int_t^{t+\varepsilon}e^{-\zeta(\hat{R}_s+Y_s-\hat{R}_t)}\left((\mu_s-r_s)\eta-\frac{1}{2}\sigma_s^2\eta^2-(\zeta+1)\sigma_s^2\hat{\pi}_s\eta \right)\md s\right]\\&-\frac{\gamma}{2}\mathbb{E}_t\left[\int_t^{t+\varepsilon}\sigma_s\eta\left(2\sigma_s\hat{\pi}_s+\sigma_s\eta\right)\md s\right].
 \end{aligned}
 $$
We first prove that 
 $$ 
        \lim\limits_{\varepsilon\rightarrow 0}\frac{1}{\varepsilon}\mathbb{E}_t\left[\int_t^{t+\varepsilon}\left|\left(e^{-\zeta\left(Y_s+\hat{R}_s-\hat{R}_t\right)}-e^{-\zeta Y_s}\right)\left((\mu_s-r_s)\eta-\frac{1}{2}\sigma_s^2\eta^2-(\zeta+1)\sigma_s^2\hat{\pi}_s\eta \right)\right|\md s\right]=0,\quad a.s..
        $$
Note that $r$, $\mu$, $\sigma$, $\eta$, $\hat{\pi}$, and $Y$ are all bounded, it suffices to show that
$$
  \lim\limits_{\varepsilon\rightarrow 0}\frac{1}{\varepsilon}\mathbb{E}_t\left[\int_t^{t+\varepsilon}\left|\left(e^{-\zeta\left(\hat{R}_s-\hat{R}_t\right)}-1\right)\right|\md s\right]=0,\quad a.s.,
$$
which follows immediately from the fact that $\{e^{-\zeta\left(\hat{R}_s-\hat{R}_t\right)}\}_{s\in[t,T]}\in L_{\mathbb{F}}^{p}\left(\Omega; \left(C[t,T];\mathbb{R}\right)\right)$.
Therefore
$$
\begin{aligned}
&\limsup\limits_{\varepsilon\rightarrow 0} \frac{1}{\varepsilon}\left(J(t,\hat{\pi}^{t,\varepsilon})-J(t,\hat{\pi})\right)\\
\leq &\limsup\limits_{\varepsilon\rightarrow 0}\frac{1}{\varepsilon}\mathbb{E}_t\left[\int_t^{t+\varepsilon} \bigg(e^{-\zeta Y_s}\Big((\mu_s-r_s)\eta-\frac{1}{2}\sigma_s^2\eta^2-(\zeta+1)\sigma_s\eta \hat{u}_s)\Big)-\frac{\gamma}{2}\sigma_s^2\eta^2-\gamma\sigma_s\eta\hat{u}_s\bigg)\md s\right]\\
=&\limsup\limits_{\varepsilon\rightarrow 0} \frac{1}{\varepsilon}\mathbb{E}_t\left[\int_t^{t+\varepsilon} \left(-\frac{1}{2}\left( e^{-\zeta Y_s}+\gamma\right)(\sigma_s\eta)^2+\sigma_s\eta\left[ e^{-\zeta Y_s}\left(\theta_s-(\zeta+1)\hat{u}_s\right)-\gamma\hat{u}_s\right]\right)\md s\right].
\end{aligned}
$$
Substituting $\hat{u}=\frac{ e^{-\zeta Y}\theta}{(\zeta+1)e^{-\zeta Y}+\gamma}$ into the above expression yields that
$$
\limsup\limits_{\varepsilon\rightarrow 0} \frac{1}{\varepsilon}\left(J(t,\hat{\pi}^{t,\varepsilon})-J(t,\hat{\pi})\right)\le\limsup\limits_{\varepsilon\rightarrow 0} \frac{1}{\varepsilon}\mathbb{E}_t\left[\int_t^{t+\varepsilon} \left(-\frac{1}{2}\left( e^{-\zeta Y_s}+\gamma\right)(\sigma_s\eta)^2\right)\md s\right]\le 0.
$$
This verifies that $\hat{\pi}$ is a time-consistent equilibrium strategy.
\end{proof}

\subsection{Existence of time-consistent equilibrium under trading constraint}\label{sec:constraint}
In this subsection, we assume that the trading strategy $\pi$ takes values in a bounded, closed  and convex  set $A\subset\mathbb{R}$. Consequently, $u=\sigma\pi$ is also bounded. To be precise, we introduce the following definition.
\begin{definition}[(Admissible trading strategy with constraint $A$)]
   A trading strategy $\pi$ is said to be admissible with constraint $A$ if $\pi_t\in A$ a.s. for any $t\in[0,T]$. We denote by $\Pi_A$ the set of all admissible trading strategies with constraint $A$.
\end{definition}
As in the proof of Theorem \ref{verification0}, it is straightforward that $\Pi_A\subseteq\Pi$.
Following \cite{yan2019time} and \cite{liang2025integral}, we next introduce the modified definitions of the perturbed strategy and the equilibrium strategy.
 \begin{definition}
    A strategy $\hat{\pi}\in\Pi_A$ is called an equilibrium strategy if
     $$   \limsup\limits_{\varepsilon\rightarrow 0}\frac{J(t,\hat{\pi}^{t,\varepsilon,a})-J(t,\hat{\pi})}{\varepsilon}\le 0 \quad a.s.,
     $$
     for any $t\in [0,T)$ and  any $a\in L_{\mathcal{F}_t}^{\infty}(\Omega,A)$, where
     $$
\hat{\pi}^{t,\varepsilon,a}_s:=\left\{\begin{aligned}
    & a, \quad s\in[t,t+\varepsilon),\\
    & \hat{\pi}_s,\quad s\in[0,t)\cup [t+\varepsilon,T]
\end{aligned}\right.
     $$
     denotes the modified perturbed strategy. 
 \end{definition}
\begin{remark}
   In this modified definition, instead of adding an $\mathcal{F}_t$-measurable random variable $\eta$ to the original strategy $\hat{\pi}$, we perturb the strategy by directly replacing it with an $\mathcal{F}_t$-measurable random variable $a$ taking values in $A$. This modification is made because requiring the perturbed strategy to remain admissible, i.e., $\pi_s + \eta \in A$ for $s \in [t, t + \varepsilon)$, would impose a strong restriction on $\eta$, possibly making it impossible to perturb many strategies.
\end{remark}
We first introduce a lemma for the projection onto a convex closed set.
\begin{lemma}\label{lma:projection}
    Let $U\subset\mathbb{R}^n$ be a convex and closed set. Fix
$w\in\mathbb{R}^n$, then we have the following two assertions: (i) if  $u=P_U(w)$ is the orthogonal
projection of $w$ onto  $U$. Let
$h\in\mathbb{R}^n$ satisfy $u+h\in U$. Then for any
$\alpha>0$ 
\begin{align}\label{eq:alpha:convex}
\bigl|\alpha(w-u)-h\bigr|\ge\alpha|w-u|,
\end{align}
and equality can occur only in the trivial case
$h=0$; (ii) if $U\ni u\neq P_U(w)$, then  for any $\alpha>0$, there exists an $h$ such that $u+h\in U$ and 
$$
\bigl|\alpha(w-u)-h\bigr|<\alpha|w-u|.
$$
Moreover, $h$ can be chosen to be $\lambda \left(P_U(w)-u\right)$ for any $\lambda\in (0,2\alpha)$.
\end{lemma}
\begin{proof}
See Appendix \ref{pflma:projection}.
\end{proof}
\begin{remark}
   In Lemma \ref{lma:projection}, we consider a convex and closed set in $\mR^n$ for any positive integer $n$, rather than in $\mR$, thereby showing that our subsequent proof remains valid in the presence of multiple risky assets.
\end{remark}
In the following discussion in this subsection, we use the similar notations as in Remark \ref{rmk:notation}. For example, we denote the perturbed strategies $\hat{\pi}^{t,\varepsilon,a}$ and $\pi^{t,\varepsilon,a}$ simply by $\hat{\pi}^{t,\varepsilon}$ and $\pi^{t,\varepsilon}$, respectively.
By an analogous argument in Section \ref{section3}, we have the following necessary condition for an equilibrium strategy.
\begin{theorem}[Necessary condition]
If $\hat{\pi}$ is an equilibrium strategy, then then there exist $(Y,\tilde{Y})$ and $(Z^1, Z, \tilde{Z}^1, \tilde{Z})$ satisfy the same condition in  Theorem \ref{nes} and the BSDE system \eqref{factor} and $\hat{\pi}$ satisfies
\begin{align} \label{eq:constraint:pi}\hat{\pi}=\sigma^{-1}\mathcal{P}_{\sigma A}\left(\frac{ e^{-\zeta Y}\theta-\zeta e^{-\zeta Y}(Z^1+\rho Z)-\gamma(\tilde{Z}^1+\rho\tilde{Z})}{(\zeta+1)e^{-\zeta Y}+\gamma}\right).
\end{align}
\end{theorem}
\begin{proof}
    The derivation here is similar to that in Theorem \ref{nes}, we only give a sketch. Let $\Delta=\sigma(a-\hat{\pi})$, by the same argument in Lemma \ref{lma:3.4}, we can obtain
 \begin{equation}\label{eq:constrnes:3.4}
\begin{aligned}
    &J(t,\hat{\pi}^{t,\varepsilon})-J(t,\hat{\pi})\\=&\mathbb{E}_t\left[\int_t^{t+\varepsilon}e^{-\zeta(\hat{R}_s+Y_s-\hat{R}_t)}(a(s,a)-a(s,\hat{\pi}_s)-\zeta\sigma_s\hat{\pi}_s\cdot \Delta_s-\zeta\Delta_s\left(\tilde{Z}^1_s+ \rho Z_s\right) )\md s\right]\\
    &-\zeta\mathbb{E}_t\left[\int_0^1 e^{-\zeta(\hat{R}_T+\lambda\hat{\xi}_T^{t,\varepsilon}-\hat{R}_t)}(1-\lambda)d\lambda|\hat{\xi}_T^{t,\varepsilon}|^2\right]\\
    &+\frac{\gamma}{2}\mathbb{E}_t\left[\int_t^{t+\varepsilon}\left((\sigma_s\hat{\pi}_s+\tilde{Z}^1_s+\rho\tilde{Z}_s)^2-(\sigma_sa+\tilde{Z}^1_s+\rho\tilde{Z}_s)^2\right)\md s\right].
\end{aligned}
    \end{equation} 
  Moreover, by the same argument in Lemma \ref{lma:3.5} and Lemma \ref{lma:3.6}, there exist a sequence $\{\varepsilon_n^t\}$  and a fixed constant $C$ independent of $a$ such that for $a.e.$ $t\in [0,T)$,
$$
\begin{aligned}
    \lim\limits_{n\rightarrow\infty} \frac{1}{\varepsilon_n^t}\mathbb{E}_t&\left[\int_t^{t+\varepsilon_n^t} e^{-\zeta(\hat{R}_s+Y_s-\hat{R}_t)}\left(a(s,a)-a(s,\hat{\pi}_s)-\zeta(Z_s^1+\sigma_s\hat{\pi}_s+\rho Z_s)\cdot \Delta_s \right)\md s\right]\\
=&e^{-\zeta Y_t}\left(\theta_t\Delta_t-\frac{1}{2}\Delta_t^2-\sigma_t\hat{\pi}_t\Delta_t-\zeta(Z_t^1+\sigma_t\hat{\pi}_t+\rho Z_t)\cdot \Delta_t\right),\quad a.s.,
\end{aligned}
$$
  
  $$
\lim\limits_{n \rightarrow\infty}\frac{1}{\varepsilon_n^t}\mathbb{E}_t\left[\int_0^1 e^{-\zeta(\hat{R}_T+\lambda\hat{\xi}_T^{t,\varepsilon}-\hat{R}_t)}(1-\lambda)d\lambda|\hat{\xi}_T^{t,\varepsilon}|^2\right]\le C\Delta_t^2,\quad a.s.,
    $$
    and 
    $$
\begin{aligned}
     \lim\limits_{n\rightarrow\infty} \frac{1}{\varepsilon_n^t}\mathbb{E}_t&\left[\int_t^{t+\varepsilon_n^t}\left((\sigma_s\hat{\pi}_s+\tilde{Z}^1_s+\rho\tilde{Z}_s)^2-(\sigma_sa+\tilde{Z}^1_s+\rho\tilde{Z}_s)^2\right)\md s\right]\\
     =&(\sigma_t\hat{\pi}_t+\tilde{Z}^1_t+\rho\tilde{Z}_t)^2-(\sigma_ta+\tilde{Z}^1_t+\rho\tilde{Z}_t)^2,\quad a.s..
\end{aligned}
    $$
Then the equilibrium strategy $\hat{\pi}$ should satisfy that 
\begin{align}\label{eq:constraint:equ}
&-\frac{e^{-\zeta Y_t}+\gamma}{2}(\sigma_ta)^2+\left(e^{-\zeta Y_t}\left(\theta_t-\zeta\left(\sigma_t\hat{\pi}_t+Z^1_t+\rho Z_t\right)\right)-\gamma\left(\tilde{Z}^1_t+\rho\tilde{Z}_t\right)\right)\sigma_ta\nonumber\\
\le & -\frac{e^{-\zeta Y_t}+\gamma}{2}(\sigma_t\hat{\pi}_t)^2+\left(e^{-\zeta Y_t}\left(\theta_t-\zeta\left(\sigma_t\hat{\pi}_t+Z^1_t+\rho Z_t\right)\right)-\gamma\left(\tilde{Z}^1_t+\rho\tilde{Z}_t\right)\right)\sigma_t\hat{\pi}_t+C\Delta_t^2,\quad a.s..
\end{align}
It follows that
\begin{equation}\label{eq:constraint:equstra}
-\frac{e^{-\zeta Y_t}+\gamma}{2}\Delta_t^2+\left(e^{-\zeta Y_t}\left(\theta_t-(\zeta+1)\sigma_t\hat{\pi}_t-\zeta\left(Z_t^1+\rho Z_t\right)\right)-\gamma\left(\tilde{Z}_t^1+\rho\tilde{Z}\right)-\gamma\sigma_t\hat{\pi}_t\right)\Delta_t\le 0 ,\quad a.s..
\end{equation}
Suppose, to the contrary, that there exists some $a\in L^{\infty}_{\mathcal{F}_t}(\Omega,A)$ such that \eqref{eq:constraint:equstra} fails to hold. Define a new admissible perturbation $a_{\lambda}=\hat{\pi}_t+\lambda\sigma_t^{-1}\Delta\in L_{\mathcal{F}_t}(\Omega,A)$  with sufficiently small constant $\lambda\in (0,\frac{\gamma}{\gamma+2C})$. On one hand, by \eqref{eq:constraint:equ}, we have
$$
\begin{aligned}
0\ge&-\frac{e^{-\zeta Y_t}+\gamma}{2}(\sigma_ta_{\lambda})^2+\left(e^{-\zeta Y_t}\left(\theta_t-\zeta\left(\sigma_t\hat{\pi}_t+Z^1_t+\rho Z_t\right)\right)-\gamma\left(\tilde{Z}^1_t+\rho\tilde{Z}_t\right)\right)\sigma_ta_{\lambda}-C|a_\lambda-\hat{\pi}_t|^2\\
 & +\frac{e^{-\zeta Y_t}+\gamma}{2}(\sigma_t\hat{\pi}_t)^2-\left(e^{-\zeta Y_t}\left(\theta_t-\zeta\left(\sigma_t\hat{\pi}_t+Z^1_t+\rho Z_t\right)\right)-\gamma\left(\tilde{Z}^1_t+\rho\tilde{Z}_t\right)\right)\sigma_t\hat{\pi}_t\\
 =&\lambda\left( -\frac{e^{-\zeta Y_t}+\gamma}{2}\lambda\Delta_t^2+\left(e^{-\zeta Y_t}\left(\theta_t-(\zeta+1)\sigma_t\hat{\pi}_t-\zeta\rho Z_t\right)-\gamma\rho\tilde{Z}-\gamma\sigma_t\hat{\pi}_t\right)\Delta_t
 -C\lambda\Delta_t^2\right) ,\quad a.s..
\end{aligned}
$$
On the other hand, as $\lambda\in (0,\frac{\gamma}{\gamma+2C})$ and \eqref{eq:constraint:equstra} does not hold for $a$, we have
$$
\begin{aligned}
    &\lambda\left( -\frac{e^{-\zeta Y_t}+\gamma}{2}\lambda\Delta_t^2+\left(e^{-\zeta Y_t}\left(\theta_t-(\zeta+1)\sigma_t\hat{\pi}_t-\zeta\rho Z_t\right)-\gamma\rho\tilde{Z}-\gamma\sigma_t\hat{\pi}_t\right)\Delta_t
 -C\lambda\Delta_t^2\right)\\
 \ge &\lambda\left(-\frac{e^{-\zeta Y_t}+\gamma}{2}\Delta_t^2+\left(e^{-\zeta Y_t}\left(\theta_t-(\zeta+1)\sigma_t\hat{\pi}_t-\zeta\rho Z_t\right)-\gamma\rho\tilde{Z}-\gamma\sigma_t\hat{\pi}_t\right)\Delta_t\right)>0
\end{aligned}
$$
on a set of positive measure, which leads to a contradiction. Therefore, \eqref{eq:constraint:equstra} must hold for all $a\in L_{\mathcal{F}_t}^{\infty}(\Omega,A)$. Consequently, we obtain
$$
\begin{aligned}
&\left|\Delta_t-\frac{e^{-\zeta Y_t}\left(\theta_t-(\zeta+1)\sigma_t\hat{\pi}_t-\zeta\left(Z_t^1+\rho Z_t\right)\right)-\gamma\left(\tilde{Z}_t^1+\rho\tilde{Z}\right)-\gamma\sigma_t\hat{\pi}_t}{e^{-\zeta Y_t}+\gamma}\right|\\
\ge & \left|\frac{e^{-\zeta Y_t}\left(\theta_t-(\zeta+1)\sigma_t\hat{\pi}_t-\zeta\left(Z_t^1+\rho Z_t\right)\right)-\gamma\left(\tilde{Z}_t^1+\rho\tilde{Z}\right)-\gamma\sigma_t\hat{\pi}_t}{e^{-\zeta Y_t}+\gamma}\right|\quad a.s..
\end{aligned}
$$

Take $\alpha_t=\frac{(\zeta+1)e^{-\zeta Y_t}+\gamma}{e^{-\zeta Y_t}+\gamma}$ and $w_t=\frac{ e^{-\zeta Y_t}\theta-\zeta e^{-\zeta Y}(Z^1_t+\rho Z_t)-\gamma(\tilde{Z}^1_t+\rho\tilde{Z}_t)}{(\zeta+1)e^{-\zeta Y_t}+\gamma}$, then the above inequality can be converted to
\begin{align}\label{eq:constraint:a}
\left|\sigma_ta-\hat{u}_t-\alpha_t\left(w_t-\hat{u}_t\right) \right|\ge \left|\alpha_t\left(w_t-\hat{u}_t\right)\right|,\quad a.s..
\end{align}
Using Lemma \ref{lma:projection}(ii), if $\hat{u}_t\neq \mathcal{P}_{\sigma_t A}\left(w_t\right)$, then, as $\alpha_t>1$,  taking $\lambda=1$ yields
$$\left|\mathcal{P}_{\sigma_t A}\left(w_t\right)-\hat{u}_t-\alpha_t(w_t-\hat{u}_t)\right|<\left|\alpha_t(w_t-\hat{u}_t)\right|.$$ 
Letting $\sigma_ta=\mathcal{P}_{\sigma_t A}\left(w_t\right)$, we get a contradiction with \eqref{eq:constraint:a}.
 Therefore, $\hat{u}_t= \mathcal{P}_{\sigma_t A}\left(w_t\right)$ a.s. and hence \eqref{eq:constraint:pi} holds.
\end{proof}
As $R$ does not appear in the BSDE system, we can, as in the case $\rho=0$, set $Z^1=\tilde{Z}^1=0$ and derive the following BSDEs:
\begin{equation}\label{BSDErho}
  \left\{
\begin{aligned}
    &\md X_s=m(s,X_s)\md t+v(s,X_s)\md \bar{B}_s, \\
    &\md Y_s=\left(\frac{\zeta}{2}|\rho Z_s+\hat{u}_s|^2+\frac{\zeta(1-\rho^2)}{2}|Z_s|^2-a(s,\hat{\pi}_s)\right)\md s+Z_s\md \bar{B}_s,\\
    &\md \tilde{Y}_s=-a(s,\hat{\pi}_s)\md s+\tilde{Z}_s\md\bar{B}_s,\\
   &  Y_T=0,\quad \tilde{Y}_T=0,\quad \hat{u}=\sigma\hat{\pi}=\mathcal{P}_{\sigma_t A}\left(\frac{ e^{-\zeta Y}\theta-\zeta e^{-\zeta Y}\rho Z-\gamma\rho\tilde{Z}}{(\zeta+1)e^{-\zeta Y}+\gamma}\right).
\end{aligned}
\right.
\end{equation}

To establish the existence of the BSDE system, we only need to verify the AB condition and the BF condition (see their definitions in Appendix \ref{secB}).
\begin{description}
    \item[AB condition:] 
    The generator of $Y$ satisfy that
    \item$$
   -r-\theta\hat{u}\le\frac{\zeta}{2}Z^2+\zeta \rho Z\hat{u}-r-\theta\hat{u}+\frac{\zeta+1}{2}|\hat{u}|^2 \le \frac{\zeta(1+\rho^2)}{2}Z^2+\frac{2\zeta+1}{2}|\hat{u}|^2-r+\theta|\hat{u}|.
    $$   The generator of  $\tilde{Y}$ satisfy that 
    $$
     -r-\theta^2\le-r-\theta \hat{u}+\frac{1}{2}|\hat{u}|^2\le \frac{1}{2}(\hat{u}-\theta)^2.
    $$  Let $a_i, i=1,2,3,4$ be $(1,0),(-1,0),(0,1),(0,-1)$. Then $a_i,i=1,2,3,4$ can positively span $\mathbb{R}^2$ and the generator satisfy the AB condition.
\end{description}
\begin{description}
    \item[BF condition]
   The generator of $Y$ satisfy that
   $$
  \left |\frac{\zeta}{2}Z^2+\zeta\rho Z\hat{u}-r-\theta\hat{u}+\frac{\zeta+1}{2}|\hat{u}|^2\right|\le \frac{\zeta(1+\rho^2)}{2}Z^2+k(t),
   $$
   with $k(t)=r(t,X_t)+\theta(t,X_t)\sup\limits_{\hat{u}\in \sigma_t A}|\hat{u}|+\frac{2\zeta+1}{2}\sup\limits_{\hat{u}\in \sigma_t A}|\hat{u}|^2$.

   Let $C_n=\sup\limits_{t\in [0,T]}\left\{|\theta(t,X_t)|+k(t)\right\}$ and
   the  generator of  $\tilde{Y}$ satisfy that 
   $$
   |-r-\theta \hat{u}+\frac{1}{2}|\hat{u}|^2|\le k(t).
   $$  
\end{description}
Here $\sup\limits_{\hat{u}\in \sigma_t A}|\hat{u}|$ is bounded with the bound depending only on $A$ and the market parameter. Then the BSDE system \eqref{BSDErho} admits a locally Holder and $\textup{BMO}$ solution (see Theorem 2.14 in \cite{xing_class_2018}). Moreover, $Y$ and $\tilde{Y}$ are bounded due to the terminal condition $Y_T=0$.

\begin{theorem}\label{verificationconstraint}
    Assume that $(Y,\tilde{Y},Z,\tilde{Z})$ is a bounded \textup{BMO} solution to the BSDE \eqref{BSDErho}, then 
    \begin{align}\label{bodeqs}
\hat{\pi}=\sigma^{-1}\hat{u}=\sigma^{-1}\mathcal{P}_{\sigma A}\left(\frac{ e^{-\zeta Y}\theta-\zeta e^{-\zeta Y}\rho Z-\gamma\rho \tilde{Z}}{(\zeta+1)e^{-\zeta Y}+\gamma}\right)
\end{align}
is an equilibrium strategy. 
 \end{theorem}
 \begin{proof}
     It is clear  that $\hat{\pi}\in \Pi_A$. We now prove that $\hat{\pi}$ is an equilibrium strategy. Fix $t\in[0,T)$ and $a\in L_{\mathcal{F}_t}^{\infty}(\Omega;A)$ and let $\Delta_s:=\sigma_s(a-\hat{\pi}_s)$.
 Then, by \eqref{eq:constrnes:3.4} and using the same arguments as in Corollary \ref{corol:verifyequilibrium} and Lemma \ref{verification0}, it holds almost surely that
$$
   \begin{aligned}
         &J(t,\hat{\pi}^{t,\varepsilon})-J(t,\hat{\pi})\\
         \le &\mathbb{E}_t\left[\int_t^{t+\varepsilon}e^{-\zeta(\hat{R}_s+Y_s-\hat{R}_t)}(a(s,a)-a(s,\hat{\pi}_s)-\zeta\sigma_s\hat{\pi}_s\cdot \Delta_s-\zeta\rho\Delta_s Z_s )\md s\right]\\
         &+\frac{\gamma}{2}\mathbb{E}_t\left[\int_t^{t+\varepsilon}((\sigma_s\hat{\pi}_s+\rho\tilde{Z}_s)^2-(\sigma_sa+\rho\tilde{Z}_s)^2)\md s\right]\\
         =& \mathbb{E}_t\left[\int_t^{t+\varepsilon}e^{-\zeta(\hat{R}_s+Y_s-\hat{R}_t)}\left(\theta_s\Delta_s-\frac{1}{2}\Delta_s^2-(\zeta+1)\sigma_s\hat{\pi}_s\Delta_s-\zeta\rho\Delta_s Z_s \right)\md s\right]\\
         &-\frac{\gamma}{2}\mathbb{E}_t\left[\int_t^{t+\varepsilon}\Delta_s\left(2\sigma_s\hat{\pi}_s+2\rho\tilde{Z}_s+\Delta_s\right)\md s\right].
         \end{aligned}
$$
    We first show that
 $$ 
        \lim\limits_{\varepsilon\rightarrow 0}\frac{1}{\varepsilon}\mathbb{E}_t\left[\int_t^{t+\varepsilon}\left|\left(e^{-\zeta\left(Y_s+\hat{R}_s-\hat{R}_t\right)}-e^{-\zeta Y_s}\right)\left(\theta_s\Delta_s-\frac{1}{2}\Delta_s^2-(\zeta+1)\sigma_s\hat{\pi}_s\Delta_s-\zeta\rho\Delta_s Z_s\right)\right|\md s\right]=0.
        $$
It suffices to prove that
$$
  \lim\limits_{\varepsilon\rightarrow 0}\frac{1}{\varepsilon}\mathbb{E}_t\left[\int_t^{t+\varepsilon}\left|\left(e^{-\zeta\left(\hat{R}_s-\hat{R}_t\right)}-1\right)Z_s\right|\md s\right]=0,\quad a.s..
$$
Note that $\{e^{-\zeta(\hat{R}_s-\hat{R}_t)}\}_{s\in[t,T]}$ satisfies
$$
\md \left(e^{-\zeta(\hat{R}_s-\hat{R}_t)}\right)=-\zeta e^{-\zeta(\hat{R}_s-\hat{R}_t)}\left(\left(a(s,\hat{\pi}_s)-\frac{\zeta}{2}\sigma_s^2\hat{\pi}_s^2\right)\md s+\sigma_s\hat{\pi}_s\md B_s\right).
$$
Because $\hat{\pi}$ is bounded, by Theorem 3.4.3 in \cite{zhang_backward_2017} and using the same argument as in the proof of Lemma \ref{diffpro}, we deduce that
$$
\mathbb{E}_t\left[\sup_{s\in[t,t+\varepsilon]}\left|e^{-\zeta(\hat{R}_s-\hat{R}_t)}-1\right|^2\right]\le C\mathbb{E}\left[\left(\int_t^{t+\varepsilon}\left|a(s,\hat{\pi}_s)-\frac{\zeta}{2}\sigma_s^2\hat{\pi}_s^2\right|\md s\right)^2+\int_t^{t+\varepsilon} \sigma_s^2\hat{\pi}_s^2\md s\right]\le C\varepsilon.
$$
Moreover, because $Z\in H_{\textup{BMO}}$, we have $\int_t^T|Z_s|^2\md s<\infty\, \,a.s.$ and   $\int_t^{t+\varepsilon}|Z_s|^2\md s\;\rightarrow\; 0 \,\, a.s.$ as $\epsilon\to 0$.  By the Conditional Dominated Convergence Theorem 
$$\mathbb{E}_t\left[\int_t^{t+\varepsilon}|Z_s|^2\md s\right]\;\longrightarrow \;0\quad a.s..$$
Then
$$
\begin{aligned}
    &\frac{1}{\varepsilon}\mathbb{E}_t\left[\int_t^{t+\varepsilon}\left|\left(e^{-\zeta\left(\hat{R}_s-\hat{R}_t\right)}-1\right)Z_s\md s\right|\right]\\
    \le &\frac{1}{\varepsilon}\left(\mathbb{E}_t\left[\int_t^{t+\varepsilon}\left|e^{-\zeta\left(\hat{R}_s-\hat{R}_t\right)}-1\right|^2\md s\right]\right)^{\frac{1}{2}}\left(\mathbb{E}_t\left[\int_t^{t+\varepsilon}Z_s^2\md s\right]\right)^{\frac{1}{2}}\\
    \le& \frac{1}{\sqrt{\varepsilon}}\left(\mathbb{E}_t\left[\sup_{s\in[t,t+\varepsilon]}\left|e^{-\zeta(\hat{R}_s-\hat{R}_t)}-1\right|^2\right]\right)^{\frac{1}{2}}\left(\mathbb{E}_t\left[\int_t^{t+\varepsilon}Z_s^2\md s\right]\right)^{\frac{1}{2}}\\
    \le& C\left(\mathbb{E}_t\left[\int_t^{t+\varepsilon}Z_s^2\md s\right]\right)^{\frac{1}{2}}\;\longrightarrow\;0,\quad a.s..
\end{aligned}
$$
 Similar to Theorem \ref{verification0}, one can derive that
 $$
 \begin{aligned}
     &\limsup\limits_{\varepsilon\rightarrow 0} \frac{1}{\varepsilon}\left(J(t,\hat{\pi}^{t,\varepsilon})-J(t,\hat{\pi})\right)\\
\le&\limsup\limits_{\varepsilon\rightarrow 0}\frac{1}{\varepsilon}\mathbb{E}_t\bigg[\int_t^{t+\varepsilon} \bigg(e^{-\zeta Y_s}\left(\theta_s\Delta_s-\frac{1}{2}\Delta_s^2-(\zeta+1)\Delta_s \hat{u}_s-\zeta\rho \Delta_s Z_s\right)\\
&\quad\quad\quad\quad\quad\quad-\frac{\gamma}{2}\Delta_s^2-\gamma\Delta_s(\hat{u}_s+\rho\tilde{Z}_s)\bigg)\md s\bigg]\\
=&\limsup\limits_{\varepsilon\rightarrow 0}\frac{1}{\varepsilon}\mathbb{E}_t\bigg[\int_t^{t+\varepsilon} \bigg(-\frac{1}{2}\left(e^{-\zeta Y_s}+\gamma\right)(\Delta_s)^2\\
&\quad\quad\quad\quad\quad\quad+\Delta_s\left[e^{-\zeta Y_s}\left(\theta_s-(\zeta+1)\hat{u}_s-\zeta\rho Z_s\right)-\gamma(\hat{u}_s+\rho\tilde{Z}_s)\right]\bigg)\md s\bigg].
 \end{aligned}
 $$
  As $\sigma>0$ and $\sigma_sA$ is a closed and convex set, by Lemma \ref{lma:projection}(i), for any $\Delta_s$ such that $\hat{u}_s+\Delta_s\in\sigma_sA$, it holds that
 $$
 \begin{aligned}
 &\left|\frac{(\zeta+1)e^{-Y_s}+\gamma}{ e^{-Y_s}+\gamma}\left(\frac{ e^{-\zeta Y_s}\theta_s-\zeta e^{-\zeta Y_s}\rho Z_s-\gamma\rho \tilde{Z}_s}{(\zeta+1)e^{-Y_s}+\gamma}-\hat{u}_s\right)-\Delta_s\right|\\
 \ge& \left|\frac{(\zeta+1)e^{-Y_s}+\gamma}{ e^{-Y_s}+\gamma}\left(\frac{ e^{-\zeta Y_s}\theta_s-\zeta e^{-\zeta Y_s}\rho Z_s-\gamma\rho \tilde{Z}_s}{(\zeta+1)e^{-Y_s}+\gamma}-\hat{u}_s\right)\right|.
 \end{aligned}
 $$
 Thus $\limsup\limits_{\varepsilon\rightarrow 0} \frac{1}{\varepsilon}\left(J(t,\hat{\pi}^{t,\varepsilon})-J(t,\hat{\pi})\right)\le0$  a.s., which completes the proof.
\end{proof}

\section{Approximate Time-Consistent Equilibrium for Small $\rho\neq 0 $}\label{sec:approximate}
The goal of this section is to construct and verify an approximate time-consistent equilibrium in the general case with small $\rho\neq 0$, where the approximation error can be shown to be the order $O(\rho^2)$. Let us first introduce the definition of an approximate time-consistent equilibrium strategy.
\begin{definition}\label{ANE}
    For $\hat{\pi}\in\Pi$,  $\hat{\pi}$ is called an approximate time-consistent equilibrium strategy with an approximation error $\Xi$ if 
     $$   \limsup\limits_{\varepsilon\rightarrow 0}\frac{J(t,\hat{\pi}^{t,\varepsilon,\eta})-J(t,\hat{\pi})}{\varepsilon}\le \Xi \quad a.s.
     $$
     for any $t\in [0,T)$ and  any $\eta\in L_{\mathcal{F}_t}^{\infty}(\Omega,\mathbb{R})$.  Moreover, if $\Xi=O(|\rho|^\alpha)$, we say that the approximation error\footnote{The order of the approximation error effectively characterizes the convergence rate of the equilibrium approximation as $\rho$ tends to $0$.} is of the order $O(|\rho|^\alpha)$.
\end{definition}

Recall that for $\rho=0$, the equilibrium strategy follows the form of $\hat{\pi}=\sigma^{-1}\frac{ e^{-\zeta Y}\theta}{(\zeta+1)e^{-\zeta Y}+\gamma}$, where $Y$ is the solution to the BSDE  \eqref{rho0}. For the case $\rho\neq 0$, we consider a trading strategy of the same structure, except that the process $Y$ is replaced by the solution $Y^{\rho}$ to the following BSDE, Specifically, let $Y^{\rho}$ and $\tilde{Y}^{\rho}$ satisfy
\begin{equation}\label{ApproBSDE}
  \left\{
\begin{aligned}
    &\md X_s=m(s,X_s)\md t+v(s,X_s)\md \bar{B}_s, \\
    &\md Y^{\rho}_s=\left(\frac{\zeta}{2}|\rho Z^{\rho}_s+u^{\rho}_s|^2+\frac{\zeta(1-\rho^2)}{2}|Z^{\rho}_s|^2-a(s,\pi^{\rho}_s)\right)\md s+Z^{\rho}_s\md \bar{B}_s,\\
    &\md \tilde{Y}^{\rho}_s=-a(s,\pi^{\rho}_s)\md s+\tilde{Z}^{\rho}_s\md\bar{B}_s,\\
   & u^{\rho}=\sigma\pi^{\rho}=\frac{ e^{-\zeta Y^{\rho}}\theta}{(\zeta+1)e^{-\zeta Y^{\rho}}+\gamma},\quad Y_T=0,\quad \tilde{Y}_T=0.
\end{aligned}
\right.
\end{equation}
Next, we will rigorously verify that the constructed strategy $\pi^{\rho}$ is an approximate time-consistent equilibrium satisfying Definition \ref{ANE}.

Let $R_t^{\rho}$ denote the log return at $t$ under  strategy $u^\rho$. We first show that the Markovian BSDE system \eqref{ApproBSDE} admits a solution such that $Y^{\rho}$ and $Z^{\rho}$ are uniformly bounded with bounds independent of $\rho \in [-1,1]$.  To this end, we need to impose an additional assumption as below.

\begin{assumption}\label{assump3}
   We  assume that
    \begin{enumerate}
        \item $v$ is differentiable w.r.t. $t,x$ and $v_x$ is uniformly bounded.
        \item    There exists a positive constant $\beta>0$ such that $v,v_x,m,r,\theta$ are Hölder continuous w. r. t $t, x$ with exponents $\beta/2, \beta$ respectively in $\mathcal{R}_T=[0,T]\times\mR$.
    \end{enumerate}
\end{assumption}

Let $\varphi(f):=\frac{e^{-\zeta f}\theta}{(\zeta+1)e^{-\zeta f}+\gamma}$ for any function $f$ on $\mathcal{R}_T$. 
Suppose that the PDE
\begin{equation}\label{eq:quasi:pde}
\left\{
\begin{aligned}
&f_t+mf_x+\frac{1}{2}v^2f_{xx}=\zeta\rho f_xv\varphi(f)+\frac{\zeta}{2}f_x^2v^2-r-\theta \varphi(f)+\frac{\zeta+1}{2}\varphi^2(f),\\
&f(T,x)=1,
\end{aligned}\right.
\end{equation}
admits a classical solution $f^{\rho}$. Given such an $f^{\rho}$, consider the linear PDE
\begin{equation}\label{eq:linear:pde:g}
    \left\{
    \begin{aligned}
     &g_t+mg_x+\frac{1}{2}v^2g_{xx}=-r-\theta \varphi(f^{\rho})+\frac{1}{2}\varphi^2(f^{\rho}), \\
     & g(T,x)=0,
    \end{aligned}\right.
\end{equation}
and suppose that it admits a classical solution $g^{\rho}$.  
 Then, by setting $Y_t^{\rho}:=f^{\rho}(t,X_t),\;\tilde{Y}_t^{\rho}=g^{\rho}(t,X_t)$ and $Z_t^{\rho}:=f^{\rho}_x(t,X_t)v(t,X_t),\;\tilde{Z}^{\rho}=g^{\rho}_x(t,X_t)v(t,X_t)$, the nonlinear Feynman--Kac formula implies that  $(Y^{\rho},\tilde{Y}^{\rho},Z^{\rho},\tilde{Z}^{\rho})$ solves the BSDE \eqref{ApproBSDE}. 
 
 Once the solution to \eqref{eq:quasi:pde} is obtained,  
the subsequent PDE~\eqref{eq:linear:pde:g} becomes linear.  
Therefore, our analysis will primarily focus on the quasi-linear PDE~\eqref{eq:quasi:pde}. 
To this end, we introduce the following classical existence theorem 
(Theorem~\ref{thm:quasilinear}) for quasi-linear parabolic equations whose 
principal part is in divergence form.  
Consider the Cauchy problem
\begin{equation}
\left\{
\begin{aligned}\label{equ:cauchy}
    &u_t - \sum_{i=1}^{n}\frac{d}{dx_i} m_i(t,x,  u, u_x) + a(t, x , u, u_x) = 0,\\
    &u(0,x)=\psi_0(x).
\end{aligned}\right.
\end{equation}
To state the theorem, we introduce the notation
$$
m_{ij}(t, x, u, p) \equiv \frac{\partial m_i(t, x, u, p)}{\partial p_j},
\quad A(t, x, u, p) \equiv a(t, x, u, p) - \frac{\partial m_i}{\partial u} p_i - \frac{\partial m_i}{\partial x_i}.
$$
\begin{theorem}\label{thm:quasilinear}
Suppose that the following conditions hold.
\begin{enumerate}
    \item[$\mathrm{a)}$] $\psi_0(x) \in C^{2+\beta}(\mR^n)$.
    
    \item[$\mathrm{b)}$] For $t \in (0, T]$ and arbitrary $x, u, p$, we have  
\begin{equation}\label{cauchy}
\mu_1\xi^2\le m_{ij}(t,x,u,p)\xi^2\le \mu_2\xi^2,
\end{equation}
with positive constants $\mu_1>0$  and $\mu_2>0$ and
    \[A(t, x, u, 0)u \geq -b_1u^2 - b_2, \quad \text{with constants}\ b_1, \, b_2 \geq 0.\]
    
    \item[$\mathrm{c)}$] For any $(t,x)\in\mathcal{R}_T=[0,T]\times\mR^n$ and $|u| \leq M$, where $M$ is a constant depending only on $b_1,b_2$ and $\psi_0$, there is a constant $\mu>0$ such that
$$
\sum_{i=1}^{n} \left( \left| m_i \right| + \left| \frac{\partial m_i}{\partial u} \right| \right) (1 + |p|) + \sum_{i, j=1}^{n} \left| \frac{\partial m_i}{\partial x_j} \right| + |a| \leq \mu (1 + |p|)^2.
$$
Moreover, the functions $m_i(t, x, u, p)$ and $a(t, x, u, p)$ are continuous, and $m_i(t,x,u,p)$ is differentiable w.r.t. variables $x,u,p$ for each $i$. 
\item[$\mathrm{d)}$]  For any $(t,x)\in\mathcal{R}_T$ and $|u| \leq M$, $|p|\le M_1$, where $M_1$ is a constant depending only on $M,\mu,\mu_1,\mu_2$ and $\psi$, the function $m_i, a, \frac{\partial m_i}{\partial p_j}, \frac{\partial m_i}{\partial u}$, and $\frac{\partial m_i}{\partial x_i}$ are continuous functions satisfying a Hölder condition in  $t, x, u$ and $p$ with exponents $\beta/2, \beta, \beta$ and $\beta$, respectively.
\end{enumerate}

Then there exists at least one solution $u(t, x)$ of the Cauchy problem \eqref{equ:cauchy} in the strip $\mathcal{R}_T$ such that $|u|\le M$, $|u_x|\le M_1$, and  $u\in C^{1+\beta/2, 2+\beta}(\mathcal{R}_T)$.
\end{theorem}
\begin{proof}
    The result follows directly from Theorem 6.1 and Theorem 8.1 in \cite{ladyzhenskaia1968linear}. The dependence of $M$ and $M_1$ can be found in Theorem 2.9 and (6.10) in \cite{ladyzhenskaia1968linear}.
\end{proof}
\begin{remark}
    The existence of solution to the BSDE system  \eqref{ApproBSDE} can be investigated in \cite{xing_class_2018} and \cite{FAN20161511}, in which a \textup{BMO} solution can be established. However, we require the boundedness of $Z$ and $\tilde{Z}$ in our verification theorem (Theorem \ref{rhoNE}). Therefore, we employ the Feynman–Kac representation and resort to a PDE-based analysis. Specifically, we utilize Theorem \ref{thm:quasilinear} as a technical tool to derive the desired boundedness.
\end{remark}
Using Theorem~\ref{thm:quasilinear}, we obtain the following result concerning the existence and uniqueness of solutions to \eqref{eq:quasi:pde} and \eqref{eq:linear:pde:g}.
\begin{proposition}\label{prop:exist:pde}
   Under Assumptions \ref{assump1}, \ref{assump2} and \ref{assump3}, there exists a solution $f$ to PDE \eqref{eq:quasi:pde} and a solution $g$ to PDE \eqref{eq:linear:pde:g}. Moreover, $f,f_x,g,g_x$ is bounded independent of $\rho$.
\end{proposition}
\begin{proof}
    We focus on the proof of PDE \eqref{eq:quasi:pde}, as PDE \eqref{eq:linear:pde:g} can be proved in a similar manner. We verify that \eqref{eq:quasi:pde} satisfies conditions (a)–(d) in Theorem \ref{thm:quasilinear}. To begin, we reverse the time by denoting $f(T-t,x)$ as  $f(t,x)$ and then rewrite the PDE in divergence form.
\begin{equation}\label{equation rho}
    f_t-\frac{\md}{\md x}\left(\frac{1}{2}v^2(t,x)f_x\right)+a(t,x,f,f_x)=0,
\end{equation}
with $f(0,x)=1$. We also denote that
{\small
$$a(t,x,f,p)=-mp+\zeta\rho pv\frac{ e^{-\zeta f}\theta}{(\zeta+1)e^{-\zeta f}+\gamma}+\frac{\zeta}{2}p^2v^2-r-\theta \frac{ e^{-\zeta f}\theta}{(\zeta+1)e^{-\zeta f}+\gamma}+\left(\frac{ e^{-\zeta f}\theta}{(\zeta+1)e^{-\zeta f}+\gamma}\right)^2+vv_xp,$$
$$A(t,x,f,p)=-mp+\zeta\rho pv\frac{ e^{-\zeta f}\theta}{(\zeta+1)e^{-\zeta f}+\gamma}+\frac{\zeta}{2}p^2v^2-r-\theta \frac{e^{-\zeta f}\theta}{(\zeta+1)e^{-\zeta f}+\gamma}+\left(\frac{ e^{-\zeta f}\theta}{(\zeta+1)e^{-\zeta f}+\gamma}\right)^2.$$}

In our case, $\psi_0=0$ and $m_{11}=\frac{1}{2}v^2$. Let $b_1=b_2=\frac{1}{2}\left(||r||_{\infty}+\frac{1}{4}||\theta||_{\infty}^2\right)$, we have
$$A(t,x,f,0)f=\left(-r-\theta \frac{ e^{-\zeta f}\theta}{(\zeta+1)e^{-\zeta f}+\gamma}+\left(\frac{ e^{-\zeta f}\theta}{(\zeta+1)e^{-\zeta f}+\gamma}\right)^2\right)f \ge -\frac{1}{2}\left(||r||_{\infty}+\frac{1}{4}||\theta||_{\infty}^2\right)(f^2+1),$$
and the first two conditions are verified. For Condition c), we have
$$
\begin{aligned}
\left( \left| m_1 \right| + \left| \frac{\partial m_1}{\partial f} \right| \right) (1 + |p|) +  \left| \frac{\partial m_1}{\partial x_j} \right| + |a|=\frac{1}{2}v^2|p|(1+|p|)+|vv_xp|+|a(t,x,f,p)|
\le C(1+|p|^2),
\end{aligned}
$$
as $\theta,v,v_x,m,\varphi(f)$ are bounded independent of $\rho$ and $-1\le\rho\le 1$. We emphasize that $C$ can be chosen independently of $\rho$. 

To verify the final condition in the theorem,  we employ a minor result concerning Hölder continuous functions. That is, if $h_1$ and $h_2$ are two bounded Hölder continuous functions with $\alpha_1$ and $\alpha_2$ being their Hölder coefficients respectively, then $h_1h_2$ is a Hölder continuous function with Hölder coefficient $\min(\alpha_1,\alpha_2)$. Then Condition d) can be easily verified and thus there exists a solution $f\in C^{1+\beta/2,\,2+\beta}(\mathcal{R}_T)$.  Moreover, as $b_1,b_2,\mu,\mu_1,\mu_2$ are independent of $\rho$, the constants $M$ and $M_1$ in Theorem \ref{thm:quasilinear} are also independent of $\rho$. 
\end{proof}
 By Proposition \ref{prop:exist:pde} and the nonlinear Feynman–Kac formula, there exists a bounded solution $(Y^{\rho},\tilde{Y}^{\rho},Z^{\rho},\tilde{Z}^{\rho})$ to the BSDE \eqref{ApproBSDE} and the processes $Z^{\rho}$ and $\tilde{Z}^{\rho}$ are bounded by a constant that does not depend on $\rho$.

Next, we turn to show that $\pi^{\rho}$ is an approximate time-consistent equilibrium.

\begin{theorem}\label{rhoNE}
    Under Assumptions \ref{assump1}, \ref{assump2}, \ref{assump3},  $\pi^{\rho}\in\Pi$ for any fixed $\rho\in[-1,1]$. Moreover, $\pi^{\rho}$ is an approximate time-consistent equilibrium with the approximate error of  the order $O(\rho^2)$ such that for any $t\in [0,T]$ and any fixed $\eta\in L_{\mathcal{F}_t}^{\infty}(\Omega;\mathbb{R})$\footnote{Recall the notations in Remark \ref{rmk:notation} and we shall also omit the dependence of the perturbed strategy on $\rho$ for notational simplicity.},
    $$
\begin{aligned}
\limsup\limits_{\varepsilon\rightarrow 0} \frac{1}{\varepsilon}\left(J(t,\pi^{t,\varepsilon})-J(t,\pi^{\rho})\right)
\le  C\rho^2.
\end{aligned}
$$
Moreover, $C$ is a positive constant independent of $t$, $\eta$ and $\rho$.
\end{theorem}
\begin{proof}
   As $\pi^{\rho}$ is  bounded, it follows from the same arguments used in the proof of Theorem \ref{verification0} that $\pi^{\rho}\in\Pi$. 
  We next show that $\pi^{\rho}$ is an approximate Nash equilibrium strategy. Fix $t\in[0,T)$ and $\eta\in L_{\mathcal{F}_t}^{\infty}(\Omega;\mathbb{R})$. By Corollary \ref{corol:verifyequilibrium} and the same line of reasoning as in Theorem \ref{verification0}, we obtain, almost surely, that
  {\small
  $$
\begin{aligned}
    J(t,\pi^{t,\varepsilon})-J(t,\pi^{\rho})
    \le& \mathbb{E}_t\left[\int_t^{t+\varepsilon} e^{-\zeta Y^{\rho}_s}\left((\mu_s-r_s)\eta-\frac{1}{2}\sigma_s^2\eta^2-(\zeta+1)\sigma_s\eta u^{\rho}_s-\zeta\rho \sigma_s\eta Z^{\rho}_s\right)\md s\right]\\&-\frac{\gamma}{2}\mathbb{E}_t\left[\sigma_s\eta(2u^{\rho}_s+2\rho\tilde{Z}^{\rho}_s+\sigma_s\eta)\md s\right].
\end{aligned}
  $$}
Using the same arguments for Theorem \ref{verification0} and Theorem \ref{verificationconstraint}, we deduce that
  
{\small
 $$ 
        \lim\limits_{\varepsilon\rightarrow 0}\frac{1}{\varepsilon}\mathbb{E}_t\left[\int_t^{t+\varepsilon}\left|\left(e^{-\zeta\left(Y^{\rho}_s+R^{\rho}_s-R^{\rho}_t\right)}-e^{-\zeta Y^{\rho}_s}\right)\left((\mu_s-r_s)\eta-\frac{1}{2}\sigma_s^2\eta^2-(\zeta+1)\sigma_s^2\pi^{\rho}_s\eta-\zeta\rho\sigma_s\eta Z^{\rho}_s\right)\right|\md s\right]=0\quad a.s..
        $$}
 Consequently, it holds that
$$
\begin{aligned}
&\limsup\limits_{\varepsilon\rightarrow 0} \frac{1}{\varepsilon}\left(J(t,\pi^{t,\varepsilon})-J(t,\pi^{\rho})\right)\\
\le&\limsup\limits_{\varepsilon\rightarrow 0}\frac{1}{\varepsilon}\mathbb{E}_t\left[\int_t^{t+\varepsilon} \bigg\{e^{-\zeta Y^{\rho}_s}\left((\mu_s-r_s)\eta-\frac{1}{2}\sigma_s^2\eta^2-(\zeta+1)\sigma_s\eta u^{\rho}_s-\zeta\rho \sigma_s\eta Z^{\rho}_s\right)\right.\\
&\left.-\frac{\gamma}{2}\sigma_s^2\eta^2-\gamma\sigma_s\eta(u^{\rho}_s+\rho\tilde{Z}^{\rho}_s)\bigg\}\md s\right] \\
=&\limsup\limits_{\varepsilon\rightarrow 0}\frac{1}{\varepsilon}\mathbb{E}_t\left[\int_t^{t+\varepsilon} \left(-\frac{1}{2}\left( e^{-\zeta Y^{\rho}_s}+\gamma\right)(\sigma_s\eta)^2-\sigma_s\eta\left[\zeta e^{-\zeta Y^{\rho}_s}\rho Z^{\rho}_s+\gamma\rho\tilde{Z}^{\rho}_s\right]\right)\md s\right]\\
\le &\limsup\limits_{\varepsilon\rightarrow 0} \frac{1}{\varepsilon}\mathbb{E}_t\left[\int_t^{t+\varepsilon} \frac{(\zeta e^{-\zeta Y^{\rho}_s} Z^{\rho}_s+\gamma\tilde{Z}^{\rho}_s)^2}{2( e^{-\zeta Y^{\rho}_s}+\gamma)}\md s\right]\rho^2\le C\rho^2.
\end{aligned}
$$
Here $C$ is a constant independent of $\rho$ and $t$.
Thus $\pi^{\rho}=\sigma^{-1}\frac{ e^{-\zeta Y^{\rho}}\theta}{(\zeta+1)e^{-\zeta Y^{\rho}}+\gamma}$ is an approximate time-consistent equilibrium with the desired approximation error $C\rho^2$.
\end{proof}
\section{Numerical Examples}\label{sec:numerical}
In this section, we present some numerical studies on the equilibrium strategy in the incomplete factor market model by employing some deep learning algorithms developed in \cite{e_deep_2017} and \cite{Han_2018} to our BSDE systems \eqref{rho0}, \eqref{BSDErho}, and \eqref{ApproBSDE}.
In particular, to fulfill our model assumptions, we consider a truncated time-varying Gaussian mean-return model 
in which the stock price process $S_t$ and the market factor process $X_t$ evolve as
$$
\begin{aligned}
\frac{\md S_t}{S_t}&=(r+\sigma \cdot\theta(X_t))\md t+\sigma\md B_t,\\
\md X_t &= \lambda(-10000\vee(\bar{X}-X_t)\wedge10000)\md t +\nu\md \bar{B}_t.
\end{aligned}
$$
Here, we choose $\theta(X):= X^+ \wedge 10000$ to ensure that Assumptions \ref{assump1}, \ref{assump2}, and \ref{assump3} are satisfied. Similar to \cite{dai_dynamic_2021}, we fix the parameter values: $r=0.017$, $\sigma=0.15$, $X_0=\bar{X}=0.273$, $\lambda=0.27$ and $\nu=0.065$.  Moreover, we set $\zeta=1$, $\gamma=0.1$ and  $T=2$ for the base model.
We discretize the horizon $[0, T]$ into time grid $0 = t_0 < t_1 < \cdots < t_N = T$, where $t_i = i\Delta t$ for $i = 0, 1, \cdots, N$ and $\Delta t = T/N$. Henceforth, for notational simplicity, we use the subscript $n$ to denote the processes evaluated at time $t_n$. The truncated Ornstein–Uhlenbeck process $X$ is approximated by the Euler scheme:
$$
X_{n+1}=X_n+\lambda(-10000\vee(\bar{X}-X_n)\wedge10000)\Delta t+\nu(\bar{B}_{n+1}-\bar{B}_{n})
$$
for $n=1,\cdots,N-1$, with $X_0=0.273$. 
To numerically illustrate our theoreticla results, we trained three models. For the decoupled BSDE systems \eqref{rho0} and \eqref{ApproBSDE}, because the equilibrium strategy $\hat{\pi}$ (for $\rho = 0$) and the approximate Nash equilibrium strategy $\pi^{\rho}$ both only depend on the process $Y$, we consider the discretized form of $Y$ in each case, respectively:
$$
Y_{n+1}=Y_n+\left(\frac{\zeta}{2}|\hat{u}_s|^2+\frac{\zeta}{2}|Z_n|^2-a(t_n,\hat{\pi}_n)\right)\Delta t+Z_n (\bar{B}_{n+1}-\bar{B}_{n}),
$$
and
$$
Y_{n+1}^{\rho}=Y_n^{\rho}+\left(\frac{\zeta}{2}|\rho Z_n^{\rho}+u^{\rho}_n|^2+\frac{\zeta(1-\rho^2)}{2}|Z_n|^2-a(t_n,\pi^{\rho}_n)\right)\Delta t+Z_n (\bar{B}_{n+1}-\bar{B}_{n}).
$$
We then employ a deep learning approximation for the initial conditions $(Y_0, Z_0)$ and the sequence $\{Z_n\}_{n=1}^{N-1}$, using $\{X_n\}_{n=0}^{N-1}$ as inputs. Particularly, we use $N-1$ fully connected feedforward neural networks with bias terms to represent $\{Z_n\}_{n=1}^{N-1}$ and two trainable parameters $(Y_0, Z_0)$ to represent the initial values of $Y$ and $Z$ at time $t_0$. Motivated by the terminal condition $Y_T = 0$, we define the mean squared loss function as
$$
Loss=\mathbb{E}\big[|Y_N|^2\big].
$$
For the coupled BSDE system \eqref{BSDErho}, we consider the following discretized counterparts:
$$
\left\{
\begin{aligned}
    &Y_{n+1}=Y_n+\left(\frac{\zeta}{2}|\rho Z_n+\hat{u}_n|^2+\frac{\zeta(1-\rho^2)}{2}|Z_n|^2-a(t_n,\hat{\pi}_n)\right)\Delta t+Z_n (\bar{B}_{n+1}-\bar{B}_{n}),\\
    &\tilde{Y}_{n+1}=\tilde{Y}_{n}-a(t_n,\hat{\pi}_n)\Delta t+\tilde{Z}_n (\bar{B}_{n+1}-\bar{B}_{n}).
\end{aligned}
\right.
$$
We adopt a similar neural network structure as above, introducing two additional parameters $(\tilde{Y}_0, \tilde{Z}_0)$ to represent the initial values of the second BSDE. Moreover, the loss function is modified to $\mathbb{E}\left[Y_T^2+\tilde{Y}_T^2\right]$.
\begin{figure}[htbp]
  \centering
  \includegraphics[width=0.8\linewidth]{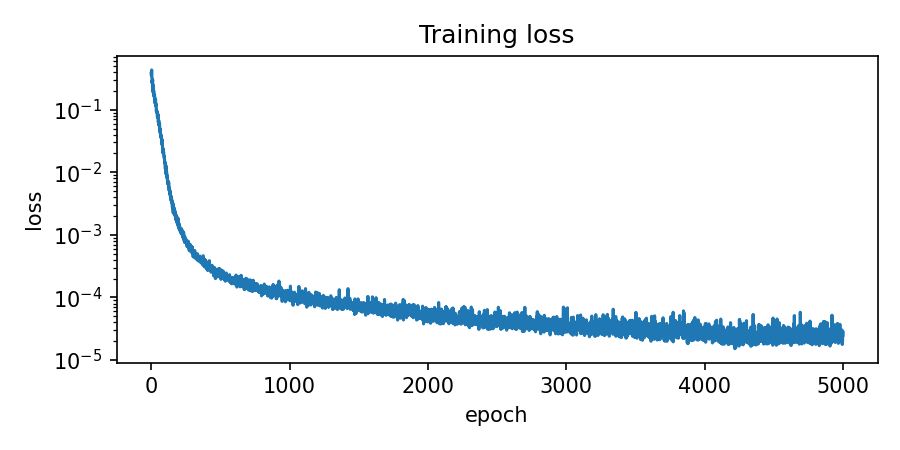}
\caption{The training loss function of an experiment with $\rho=0$.}
  \label{fig:loss_curve}
\end{figure}

Note that each neural network has the same structure: one input layer, two hidden layers, and one output layer. Both the input and output layers are 1 dimensional, while each hidden layer has 11 dimensions. We adopt batch normalization immediately after each matrix multiplication and before applying the ReLU activation function. At each epoch, we use the Adam optimizer to update the parameters with mini-batches of $512$ samples.   
Figure \ref{fig:loss_curve} plots the mean of the loss function for the model with $\rho = 0$ in the experiment. It is observed that the average loss decreases to approximately $3 \times 10^{-5}$ after 5,000 epochs.\footnote{ Given the small variation range of the learned $Y_t$ (see Figure \ref{fig:trajectories}), a sufficiently low loss is required for good relative accuracy. We employ a four-stage linear decay learning rate schedule to balance exploration and convergence:
(1) exploration: $8\times10^{-4} \to 5\times10^{-4}$;
(2) convergence: $5\times10^{-4} \to 2\times10^{-4}$;
(3) fine-tuning: $2\times10^{-4} \to 5\times10^{-5}$;
(4) terminal reinforcement: $5\times10^{-5} \to 1\times10^{-5}$.}

 \begin{figure}[htbp]
  \centering
  \includegraphics[width=\linewidth]{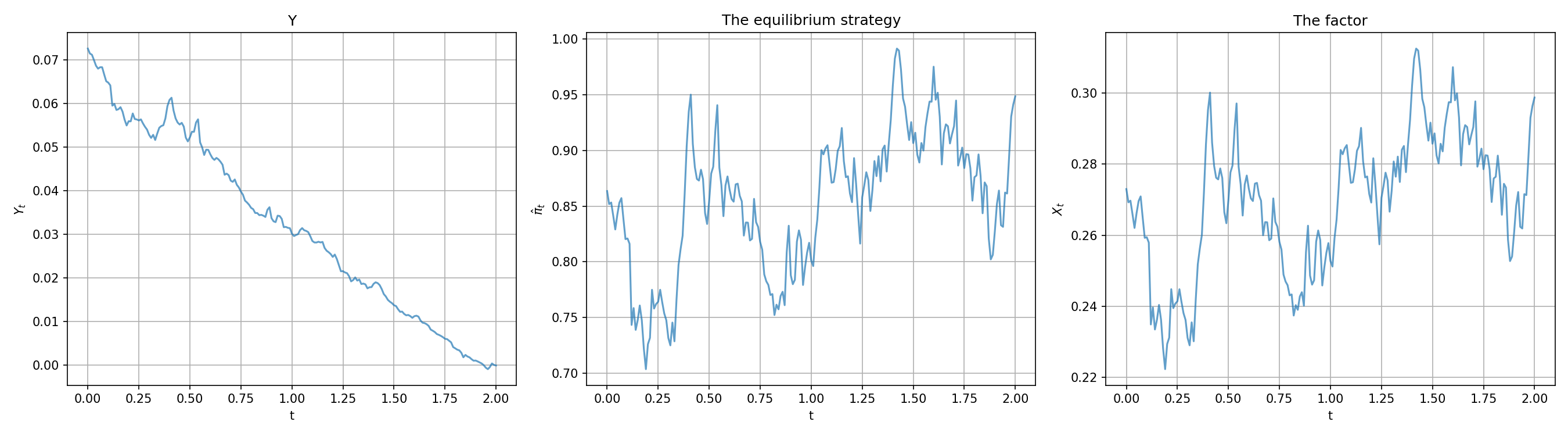}
\caption{A sample trajectory of $Y$, $\hat{\pi}_t$ and $X_t$ with $\rho=0$}
  \label{fig:trajectories}
\end{figure}

Figure \ref{fig:trajectories} presents a sample trajectory of the learned solution $Y_t$, the associated equilibrium strategy $\hat{\pi}_t$ and the factor process $X_t$. It is observed that the sample path of $Y$ exhibits an overall downward trend over time. 
In our framework, the process $Y$ satisfies
 \begin{align*}
 Y_s&=-\frac{1}{\zeta}\log(e^{-\zeta Y_s})=-\frac{1}{\zeta}\log\left(\mathbb{E}_s\left[e^{-\zeta(R_T-R_s)}\right]\right)\\
 &=-\frac{1}{\zeta}\log\left(-\mathbb{E}_s\left[U(R_T-R_s)\right]\right)=-\frac{1}{\zeta}\log\left(\mathbb{E}_s\left[\frac{1}{\zeta}U^{\prime}(R_T-R_s)\right]\right).
 \end{align*}
 Intuitively, $\mathbb{E}_s\left[U(R_T-R_s)\right]$ decreases as time evolves, leading to a decrease in $Y_s$. Therefore, $\mathbb{E}_s\left[U^{\prime}(R_T-R_s)\right]$ increases as $Y_s$ declines. As the investment horizon approaches, the expected marginal utility rises; consequently, the investor becomes more inclined to invest.
 This also explains why the equilibrium strategy $\hat{\pi}$ in the mean–variance problem in \cite{dai_dynamic_2021} remains constant when $\rho=0$, regardless of the volatility of the factor $X_t$ because the expected marginal utility is constant in that case.
    
 \begin{figure}[htbp]
     \centering
     \includegraphics[width=\linewidth]{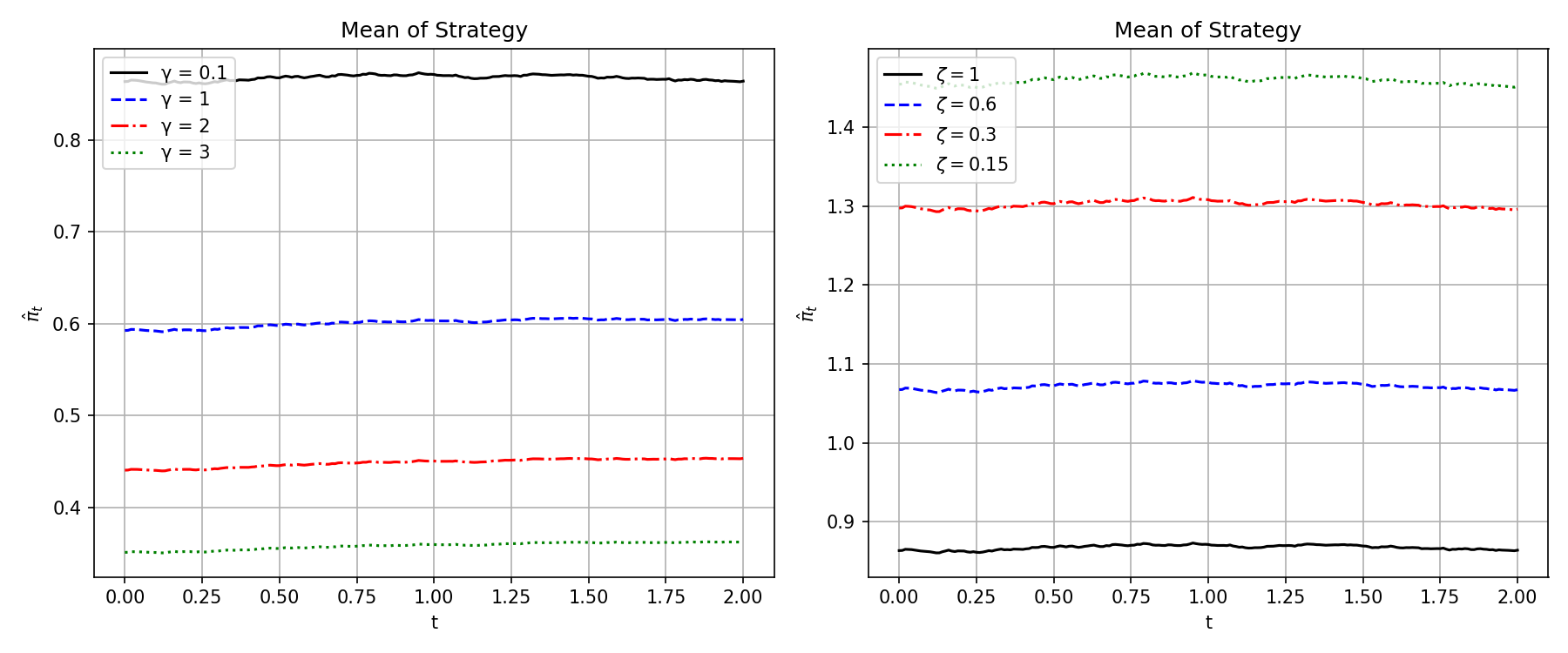}
     \caption{Plots of the mean of the equilibrium strategy $\hat{\pi}$ of $1000$ trajectories under different $\gamma$ (left panel) and different $\zeta$ (right panel) with $\rho=0$.}
     \label{fig:comparegammazeta}
 \end{figure}
 
 Figure \ref{fig:comparegammazeta} plots the mean of the strategy $\hat{\pi}$ using $1000$ trajectories under different values of $\gamma$ and $\zeta$. As illustrated, a higher $\gamma$ corresponds to a smaller investment proportion $\hat{\pi}$, which is intuitively reasonable. This behavior reflects that a larger $\gamma$ indicates the agent places a higher weight on the variance term, leading the investor to favor stability over high returns. Moreover, a higher $\zeta$ also leads to a smaller investment proportion $\hat{\pi}$, which is also intuitively reasonable as $\gamma$ represents the constant risk aversion in the CARA utility function..

 \begin{figure}[htbp]
     \centering
     \includegraphics[width=\linewidth]{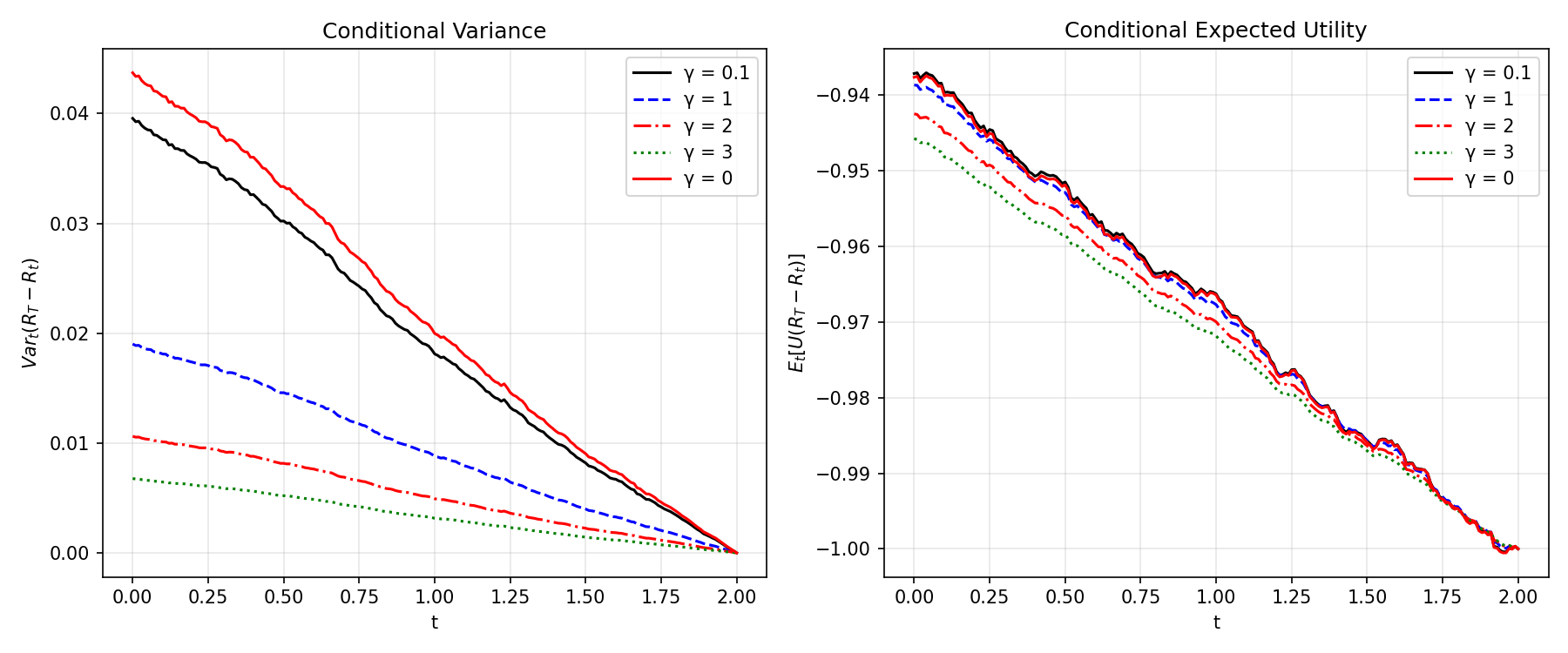}
     \caption{Comparison of conditional variance and conditional expected utility: our model with different $\gamma$ vs. model without variance term}
     \label{fig:compare_variance_utility}
 \end{figure}
Figure \ref{fig:compare_variance_utility} compares the conditional variance and expected utility under the equilibrium strategy $\hat{\pi}$ in our model with $\rho=0$ and the corresponding strategy in the same market setting where the agent ignores the variance term. The optimal strategy in this case can be obtained by setting $\gamma=0$:
$$
\pi = \sigma^{-1}\cdot\frac{\theta(X_t)}{2}.
$$
We then simulate 1,000 trajectories of the log-return process under each strategy and compute the sample mean and variance of $R_T - R_t$ at any time $t = n \Delta t$, $n = 0,1,2,\cdots,N$. It can be observed that, by incorporating the variance term in the $t$-functional, our equilibrium strategy reduces the conditional variance at the cost of a very small decrease in conditional expected utility.

 \begin{figure}[htbp]
     \centering
     \includegraphics[width=\linewidth]{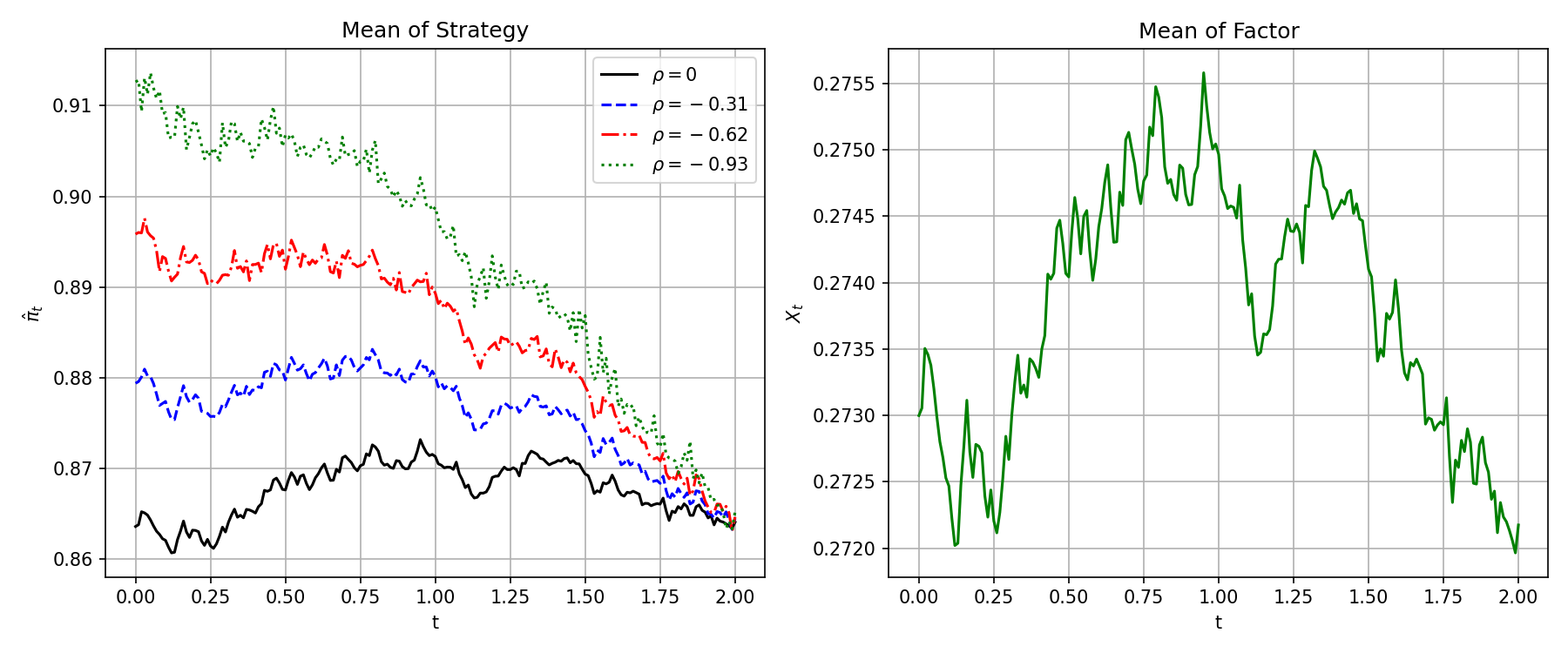}
     \caption{The mean of $\hat{\pi}_t$ and the factor $X_t$ under trading constraint with different $\rho$.}
     \label{fig:Comparerho}
 \end{figure}
 Figure \ref{fig:Comparerho} illustrates the equilibrium strategies with trading constraints $A=[-10000,10000]$ and negative correlations ($\rho=-0.31, -0.62, -0.93$) alongside the equilibrium strategy with $\rho=0$. Although the differences in $Y_t$ are small, there is a clear gap between the strategies. The constrained cases with negative $\rho$ exhibit higher investment levels. Furthermore, the investment decreases as the terminal time approaches,  which is consistent with the results in \cite{dai_dynamic_2021}.

\begin{table}[h!]
\centering
\begin{tabular}{cccc}
\toprule
\textbf{Time} & \textbf{Approximation ($\hat{\pi}^{\rho}$)}& \textbf{Constraint ($\hat{\pi}$)}& \textbf{Relative Error (‰)} \\
\midrule
0.0& -0.959920 & -0.960158 & 0.248‰ \\
0.5& -0.968967 & -0.969083 & 0.121‰ \\
1.0& -0.977834 & -0.977824 & -0.010‰ \\
1.5& -0.990868 & -0.990840 & -0.028‰ \\
2.0& -1.000000 & -1.000000 & -0.000‰ \\
\bottomrule
\end{tabular}
\caption{Comparison of Approximation and Constrained Equilibrium Strategies (six significant digits)}\label{tab:compare appro and constr}
\end{table}

 Table \ref{tab:compare appro and constr} compares the approximate Nash equilibrium strategy $\hat{\pi}^{\rho}$ with the equilibrium strategy $\hat{\pi}$  under a trading constraint at five discrete time point $(t=0.0,\;0.5,\;1.0,\;1.5,\;2.0)$ with $\rho=-0.31$ and $\gamma=1$ fixed\footnote{ In this case, because the loss from the variance term is relatively small, we use $\gamma=1$ instead of $0.1$ to make the numerical results more intuitive. }.  The last column reports the relative errors defined as $\left|
\frac{\hat{\pi}^{\rho}(t) - \hat{\pi}(t)}{\hat{\pi}(t)}
\right|$ in per mille (‰). The results show that the discrepancies between the two strategies are negligible, and the relative errors remain extremely small across all time points.  Furthermore, as $t\rightarrow T$, the relative error exhibits a slightly  decreasing trend. This indicates that the expected payoff generated by the approximate Nash equilibrium strategy closely matches that of the true constrained equilibrium.
\vskip 10pt

\section{Concluding Remarks}\label{sec:conclusion}
This paper investigates a time-inconsistent portfolio optimization problem that integrates expected utility maximization with variance penalization of log returns in an incomplete market. For the general correlated market—where direct analysis becomes intractable—we construct an approximate time-consistent equilibrium. To the best of our knowledge, there are only a few results on time inconsistency in incomplete markets, and we briefly summarize four representative examples.

The work \cite{yan2019open} investigates the mean–variance problem under stochastic volatility. Another example, mentioned in the introduction, is \cite{Gu_2020}, which studies the equilibrium strategy for a utility–deviation risk control problem. A third example is \cite{liang2025dynamic}, where the authors analyze equilibrium strategies under weighted utilities. Yet another related work is \cite{mbodji2025portfolio}, which examines a portfolio selection problem in a stochastic factor model with non-constant time discounting, focusing on the special case $\rho=1$.

These studies adopt different methodologies tailored to their respective model structures, and the techniques are typically not directly transferable across problems, reflecting the substantial difficulty of analyzing time inconsistency in incomplete markets. Our discussion on the approximate equilibrium aims to provide a more flexible and unified framework, which may offer a new perspective for handling a broader class of time inconsistent control problems in general market models and motivate some relevant future studies.
\vskip 10pt
\noindent
\textbf{Acknowledgements}:  Zongxia  Liang  is supported by the National Natural Science Foundation of China under grant no. 12271290. Sheng Wang acknowledges Professor Ka Chun Cheung and the financial supports as a postdoctoral fellow from Department of Statistics and Actuarial Science, School of Computing and Data Science, The University of
Hong Kong. Part of this work was completed in 2024, when Sheng Wang was a visiting student (research assistant) under the supervision of Professor Xiang Yu. He  acknowledges the financial support under the Hong Kong Polytechnic University research grant under no. P0045654. Xiang Yu is supported by the Hong Kong RGC General Research Fund (GRF) under grant no. 15211524, the Hong Kong Polytechnic University research grant under no. P0045654 and the Research Centre for Quantitative Finance at the Hong Kong Polytechnic University under grant no. P0042708.

 {\small
 \setlength{\bibsep}{2pt}
\bibliographystyle{abbrvnat} 
\bibliography{reference}
}

 \appendix
 \section{The Energy Inequality and the John-Nirenberg Inequality}\label{secA}
 For the reader's convenience, we collect below the energy inequality and the John-Nirenberg inequality  (see \cite{kazamaki_continuous_1994} and \cite{zhang_backward_2017}). 
\begin{lemma}\label{lma:bmo:property}
    Assume that $Z\in H_{\textup{BMO}}$.   For each integer $n\ge 1$, it holds that
    $$
    \mathbb{E}\left[\left(\int_0^T Z^2_s\md s\right)^n\right]\le n!\|Z\|_{\textup{BMO}}^{2n}.
    $$
    Moreover, the conditional version of the energy inequality also holds that, for any stopping time $\tau$,
    $$
    \mathbb{E}_{\tau}\left[\left(\int_{\tau}^T Z^2_s\md s\right)^n\right]\le n!\|Z\|_{\textup{BMO}}^{2n} \quad a.s..
    $$
\end{lemma}

If further $\|Z\|_{\textup{BMO}}< 1$, we have the John-Nirenberg inequality. 
\begin{lemma}\label{John-Nirenberg}
    If $\|Z\|_{\textup{BMO}}< 1$, for every stopping time $\tau$,
    $$
    \mathbb{E}_{\tau}\left[\exp\left(\int_{\tau}^T Z_s^2\md s\right)\right]\le \frac{1}{1-\|Z\|_{\textup{BMO}}^2}\quad a.s..
    $$
\end{lemma}
  
\section{Sufficient Conditions for the Existence of Solution to BSDEs}\label{secB}

For the reader's convenience, we present a sufficient condition in \cite{xing_class_2018} for the existence of a solution to the following BSDE system. 
\begin{equation}
\label{eq:AppendixBSDE}
dY_t = -f(t, X_t, Y_t, Z_t)\,dt + Z_t \sigma(t, X_t)\,dW_t, 
\qquad Y_T = g(X_T),
\end{equation}
where $f : [0,T] \times \mathbb{R}^d \times \mathbb{R}^N \times \mathbb{R}^{N \times d} \to \mathbb{R}^N,\; g : \mathbb{R}^d \to \mathbb{R}^N$ are Borel functions.
We first introduce two structural conditions on the generator $f$ that appear in the theorem. 

\begin{definition}[A priori boundedness (AB) condition]\label{def:AB}
We say that a function 
\[
f:[0,T]\times \mathbb{R}^d\times \mathbb{R}\times \mathbb{R}^{N\times d}\to \mathbb{R}^N
\]
satisfies the \emph{condition AB} if there exist a deterministic function 
$l\in L^1([0,T])$ and a set of vectors $a_1,\cdots,a_K\in\mathbb{R}^N$ 
positively spanning $\mathbb{R}^N$ such that
\begin{equation}\label{eq:AB}
    a_k^\top f(t,x,y,z) \le l(t) + \frac12 \lvert a_k^\top z\rvert^2,
    \qquad \text{for all } (t,x,y,z) \text{ and } k=1,\cdots,K.
\end{equation}

We say that $f$ satisfies the \emph{weak  AB condition} (short as $(wAB)$) if 
there exist Borel functions 
\[
L_k : [0,T]\times \mathbb{R}^d\times \mathbb{R}^{N\times d} \to \mathbb{R}^d,
\quad k=1,\cdots,K,
\]
such that $\lvert L_k(t,x,z)\rvert \le C(1+\lvert z\rvert)$ for some constant $C$, and
\begin{equation}\label{eq:wAB}
    a_k^\top f(t,x,y,z) \le l(t) + \frac12 \lvert a_k^\top z\rvert^2
      + a_k^\top z\, L_k(t,x,z),
    \qquad \text{for all } (t,x,y,z) \text{ and } k=1,\cdots,K.
\end{equation}
\end{definition}
\begin{remark}
    A set of nonzero vectors
$a_1, \cdots, a_K$ in $\mathbb{R}^N$ (with $K > N$) is said to positively span
$\mathbb{R}^N$ if, for each $a \in \mathbb{R}^N$, there exist nonnegative constants
$\lambda_1, \cdots, \lambda_K$ such that
\[
    \lambda_1 a_1 + \cdots + \lambda_K a_K = a.
\]
\end{remark}
\begin{definition}[The Bensoussan--Frehse (BF) condition]\label{def:BF}
We say that a continuous function 
\[
f : [0,T] \times \mathbb{R}^d \times \mathbb{R}^N \times \mathbb{R}^{N\times d} \to \mathbb{R}^N
\]
satisfies the \emph{BF} condition if it admits a decomposition of the form
\begin{equation}\label{eq:BF-decomposition}
f(t,x,y,z) = \mathrm{diag}(z\, l(t,x,y,z)) + q(t,x,y,z) + s(t,x,y,z) + k(t,x),
\end{equation}
such that the functions 
\[
l : [0,T] \times \mathbb{R}^d \times \mathbb{R}^N \times \mathbb{R}^{N\times d} \to \mathbb{R}^{d\times N}
\quad\text{and}\quad
q,s,k : [0,T] \times \mathbb{R}^d \times \mathbb{R}^N \times \mathbb{R}^{N\times d} \to \mathbb{R}^N
\]
fulfill the following conditions: there exist $b_0 \in \mathbb{R}^d$ and two sequences $\{C_n\}$ and $\{q_n\}$ of positive constants with $q_n > 1 + d/2$, and a sequence $\{\kappa_n\}$ of functions 
\[
\kappa_n : [0,\infty) \to [0,\infty), \qquad \lim_{w\to\infty}\frac{\kappa_n(w)}{w^2} = 0,
\]
such that for each $n\in\mathbb{N}$ and any $(t,x,y,z) \in [0,T] \times B_n(b_0) \times \mathbb{R}^N \times \mathbb{R}^{N\times d}$, we have
\begin{align*}
|l(t,x,y,z)| &\le C_n(1 + |z|), && \textnormal{(quadratic-linear)} \\
|q^i(t,x,y,z)| &\le C_n\left( 1 + \sum_{j=1}^i |z^j|^2 \right),
&& i=1,\cdots,N, \qquad \textnormal{(quadratic-triangular)} \\
|s(t,x,y,z)| &\le \kappa_n(|z|), && \textnormal{(subquadratic)} \\
k &\in L^{q_n}([0,T]\times B_n), && \textnormal{(z-independent)}.
\end{align*}
In this case, we write $f \in \mathrm{BF}(\{C_n\},\{\kappa_n\},\{q_n\})$.
\end{definition}

\begin{theorem}[Existence under BF + AB conditions]\label{thm:existence_BF_AB}
Suppose that $f$ satisfies \textnormal{BF} and \textnormal{AB} conditions, 
and that $g \in C^{\{\alpha_n\}}_{\mathrm{loc}, b_0}$ for some $b_0$ and it satisfies $\lim_{|x| \to \infty} \frac{|g(x)|}{|x|^2} = 0$. Then the system \eqref{eq:AppendixBSDE} admits a locally Hölderian solution $(v,w)$, 
that is, $v \in C^{\{\alpha_n'\}}_{\mathrm{loc}, b_0}$ for some sequence 
$\{\alpha_n'\}$ in $(0,1]$. When $g$ is bounded, the  \textnormal{AB} condition 
can be replaced by \textnormal{wAB} condition and $(v,w)$ is a bounded 
\textnormal{BMO}-solution.
\end{theorem}

 \section{The Challenge of our BSDE System when $\rho\neq0$}\label{secC}
  For the general case $\rho\neq 0$, we encounter a system of coupled two-dimensional quadratic BSDEs. Existing results for the existence of solution to the multidimensional BSDEs bifurcate into two paradigms: fixed-point arguments via contraction mappings (yielding the unique solution) (see \cite{FAN2023105}, \cite{Luo_2017} and \cite{HU20161066}) and construction of uniformly convergent approximants through coefficient regularization (see  \cite{xing_class_2018}, \cite{Jackson_2022} and \cite{JACKSON2023}). In this section, we explain that these two methods are not applicable in our setting. 
   
    For the coupled BSDE system in \eqref{simplified BSDE}, the drift term $f=(f_1,f_2)$ can be computed as
    $$
    \begin{aligned}
    f_1=-\left(\frac{\zeta}{2}|\rho Z+\hat{u}|^2+\frac{\zeta(1-\rho^2)}{2}|Z|^2-a(t,\hat{\pi})\right)
        = h_1+g_1
    \end{aligned}
    $$ 
   where $h_1$ is the quadratic term in $f_1$  w.r.t $z$ and $\tilde{z}$ 
$$
\begin{aligned}
h_1=-\frac{\zeta}{2}z^2+\rho\zeta z\cdot \frac{\zeta e^{-\zeta y }\rho z+\gamma\rho \tilde{z}}{(\zeta+1)e^{-\zeta y}+\gamma}-\frac{\zeta+1}{2}\left(\frac{\zeta e^{-\zeta y }\rho z+\gamma\rho \tilde{z}}{(\zeta+1)e^{-\zeta y}+\gamma}\right)^2
\end{aligned}
$$
and $g_1$ collects the linear and constant terms
$$
\begin{aligned}
g_1=&-\zeta\rho z \frac{e^{-\zeta y }\theta(t,x)}{(\zeta+1)e^{-\zeta y}+\gamma}+r(t,x)+\theta(t,x)\cdot\frac{e^{-\zeta y }\theta(t,x)-\zeta e^{-\zeta y }\rho z-\gamma\rho \tilde{z}}{(\zeta+1)e^{-\zeta y}+\gamma}\\
&-\frac{\zeta+1}{2}\left(\frac{e^{-\zeta y }\theta(t,x)}{(\zeta+1)e^{-\zeta y}+\gamma}\right)^2+(\zeta+1)\frac{e^{-\zeta y }\theta(t,x)}{(\zeta+1)e^{-\zeta y}+\gamma}\cdot \frac{\zeta e^{-\zeta y }\rho z+\gamma\rho \tilde{z}}{(\zeta+1)e^{-\zeta y}+\gamma}.
\end{aligned}
$$
Similarly, $
f_2=a(t,\hat{\pi})=h_2+g_2$,
where $h_2$ is the quadratic term in $f_2$ w.r.t. $z$ and $\tilde{z}$ 
$$
\begin{aligned}
   h_2=-\frac{1}{2}\left(\frac{\zeta e^{-\zeta y }\rho z+\gamma\rho \tilde{z}}{(\zeta+1)e^{-\zeta y}+\gamma}\right)^2.
\end{aligned}
$$
and the linear and constant terms are given by
$$
\begin{aligned}
g_2=&r(t,x)+\theta(t,x)\cdot\frac{e^{-\zeta y }\theta(t,x)-\zeta e^{-\zeta y }\rho z-\gamma\rho \tilde{z}}{(\zeta+1)e^{-\zeta y}+\gamma}\\
&-\frac{1}{2}\left(\frac{e^{-\zeta y }\theta(t,x)}{(\zeta+1)e^{-\zeta y}+\gamma}\right)^2+\frac{e^{-\zeta y }\theta(t,x)}{(\zeta+1)e^{-\zeta y}+\gamma}\cdot \frac{\zeta e^{-\zeta y }\rho z+\gamma\rho \tilde{z}}{(\zeta+1)e^{-\zeta y}+\gamma}.
\end{aligned}
$$

The first approach typically requires a continuity condition in $y$ ensuring that
$$
|f_i(t,x,y_1,z,\tilde{z})-f_i(t,x,y_2,z,\tilde{z})|\le C\phi\left(y_1\vee y_2\right)\left(1+||(z,\tilde{z})||\right)|y_1-y_2|,\quad i=1,2,
$$
so that a contraction mapping applies.
However, this condition is violated in our setting because
 $y$ appears inside the quadratic term of $z,\tilde{z}$.

The second approach relies on the AB condition or the wAB condition (see Appendix \ref{secB}) to guarantee the convergence of a subsequence of approximating solutions to the limiting BSDE system.
Specifically, assume that the wAB condition holds. Then there exist a deterministic function $l\in L^1[0,T]$, vectors $\{a_1,\cdots, a_K\}$ positively spanning $\mathbb{R}^2$, and functions $L_k : [0,T]\times \mathbb{R}\times \mathbb{R}^2 \to \mathbb{R},
\; k=1,\cdots,K$, satisfying $L_k\le C(1+|z|)$ such that
$$
a_k^{\top}f(t,x,y,z)\le l(t)+\frac{1}{2}|a_k^{\top}z|^2+a_k^\top z\, L_k(t,x,z), \quad k=1,\cdots ,K.
$$
Since $\{a_1,\cdots, a_K\}$ positively span $\mathbb{R}^2$, there  exist  $a_i,\;a_j\in \{a_1,\cdots, a_K\}$ such that $a_i\cdot (-1,0)>0$ and $a_j\cdot (0,-1)>0$. Without loss of generality, take $a_i=a_1=(-a,b)$ and $a_j=a_2=(c,-d)$ with $a>0,\;d>0$. Then, for any $(t,x,y,z,\tilde{z}) $,
\begin{align}
-af_1+bf_2&\le l(t)+\frac{1}{2}|-az+b\tilde{z}|^2+(-az+b\tilde{z}) L_2(t,x,z,\tilde{z}),\label{eq:challenge1}\\
cf_1-df_2&\le l(t)+\frac{1}{2}|cz-d\tilde{z}|^2+(cz-d\tilde{z}) L_2(t,x,z,\tilde{z}),\label{eq:challenge2}
\end{align}
where $L_i(t,x,z,\tilde{z})\le C(1+|z|+|\tilde{z}|),\;i=1,2$ for a constant $C$.  

If $b=0 \;(c=0)$,  the inequality \eqref{eq:challenge1}(\eqref{eq:challenge2}) fails to hold as $\tilde{z}\;(z)\rightarrow \infty$ and $z\;(\tilde{z})=0$ because $\rho\neq0$\footnote{Only when $\rho=0$ can vectors of the form $(-a,0)$ and $(0,-d)$  belong to $\{a_1,\cdots,a_k\}$; see Section \ref{sec:constraint}.}. 
We now assume that $b\neq 0 $ and $c\neq 0$.
Then \eqref{eq:challenge1} and \eqref{eq:challenge2} imply
\begin{equation}\label{eq:challenge3}
    -af_1+bf_2\le l(t),\quad \textup{for any}\; (t,x,y,z,\tilde{z}) \;\textup{such that}\; az=b\tilde{z},
\end{equation}
and
\begin{equation}\label{eq:challenge4}
    cf_1-df_2\le l(t) ,\quad \textup{for any}\; (t,x,y,z,\tilde{z}) \;\textup{such that}\; cz=d\tilde{z}.
\end{equation}

By substituting $z = \frac{b}{a}\tilde{z}$, the left-hand side of the inequality \eqref{eq:challenge3} becomes a quadratic polynomial in $\tilde{z}$, while the right-hand side is a deterministic function independent of $z$ and $\tilde{z}$. A necessary condition for the inequality to hold is that the quadratic coefficient of the polynomial is non-positive, namely
$$
-ah_1+bh_2\le 0 ,\quad \textup{for any}\; (t,x,y,z,\tilde{z}) \;\textup{such that}\; az=b\tilde{z},
$$
 Moreover, letting $y\rightarrow\infty$ and $y\rightarrow-\infty$ yields
$$
\frac{\zeta e^{-\zeta y }\rho z+\gamma\rho \tilde{z}}{(\zeta+1)e^{-\zeta y}+\gamma}\longrightarrow \rho\tilde{z},\quad \frac{\zeta e^{-\zeta y }\rho z+\gamma\rho \tilde{z}}{(\zeta+1)e^{-\zeta y}+\gamma}\longrightarrow \frac{\zeta\rho z}{\zeta+1}.
$$
Without loss of generality, set $a=1$. Substituting it into $h_1$ and $h_2$ with $z=b\tilde{z}$ gives 
 \begin{align}\label{C1}
\frac{\zeta}{2}b^2-\left(\rho^2\zeta+\frac{\rho^2}{2}\right)b+\frac{\zeta+1}{2}\rho^2\le 0 , \quad \left(\frac{\zeta}{2}-\frac{\zeta^2\rho^2}{2(\zeta+1)}\right)b^2-\frac{\rho^2\zeta^2}{2(\zeta+1)^2}b^3\le 0.
 \end{align}
Likewise, from \eqref{eq:challenge4} we obtain
$$
  ch_1-dh_2\le 0 ,\quad \textup{for any}\; (t,x,y,z,\tilde{z}) \;\textup{such that}\; cz=d\tilde{z}.
$$
Setting $d=1$ and substituting $\tilde{z}=cz$, then  letting $y\to \pm\infty$, we get
\begin{align}\label{C2}
\frac{1}{2}\rho^2 c^2-c\left(\frac{\zeta}{2}-\rho^2\zeta c+\frac{\zeta+1}{2}\rho^2 c^2\right)\le 0,\quad \frac{\zeta^2\rho^2}{2(\zeta+1)^2}-c\left(\frac{\zeta}{2}-\frac{\zeta^2\rho^2}{2(\zeta+1)}\right)\le 0.
\end{align}
From \eqref{C1}, we deduce
 \begin{align*}
 \left\{
 \begin{aligned}
&\rho^4(\zeta+\frac{1}{2})^2-\rho^2\zeta(\zeta+1)\ge 0,\\
&\frac{\rho^2\zeta+\frac{\rho^2}{2}-\sqrt{\rho^4(\zeta+\frac{1}{2})^2-\rho^2\zeta(\zeta+1)}}{\zeta}\le b\le \frac{\rho^2\zeta+\frac{\rho^2}{2}+\sqrt{\rho^4(\zeta+\frac{1}{2})^2-\rho^2\zeta(\zeta+1)}}{\zeta}\le \frac{\zeta+1}{\zeta},\\
&\frac{\zeta+1}{\zeta}\le\frac{(\zeta+1)^2}{\rho^2\zeta}-(\zeta+1)\le b,
\end{aligned}
\right.
 \end{align*}
which result in $\rho=1$ and $b=\frac{\zeta+1}{\zeta}>0$. Substituting $\rho=1$ into \eqref{C2} then yields $
 c\ge\frac{\zeta}{\zeta+1}$.

We have therefore shown that if $\rho=1$, then any vector $a_i=(m,n)\in \{a_1,\cdots,a_k\}$ selected in the wAB condition must satisfy one of the following:
\begin{enumerate}
    \item $m\ge0$, $n\ge 0$.
    \item $m<0$, $n=-\frac{\zeta+1}{\zeta}m$.
    \item $n<0$, $m\ge - \frac{\zeta}{\zeta+1}n$.
\end{enumerate}
 Each of the three cases implies that $(m,n)\cdot (-1,-\frac{\zeta}{\zeta+1})\le 0$.  Moreover, if $\rho\neq1$, then for any $a_i=(m,n)\in \{a_1,\cdots,a_k\}$ selected in wAB condition, we have  $m\ge 0$, which further implies $(m,n)\cdot (-1,0)\le 0$.
 As a consequence, $\{a_1,\cdots,a_k\}$ cannot positively span $\mathbb{R}^2$. Hence, neither the AB condition nor the wAB condition can hold in our setting.

 \section{Technical proofs in section \ref{sec:formulation}, section \ref{section3} and section \ref{sec:factor}}\label{sec:D}
  In this appendix, we collect proofs of Lemma \ref{lma:2.4}, Lemma \ref{diffpro}, Corollary \ref{coro:admi:piepsilon}, Lemma \ref{lma:3.4}, Corollary \ref{corol:verifyequilibrium}, Lemma \ref{lma:3.5}, Lemma \ref{lma:3.6} and Lemma \ref{lma:projection}.
 
  \subsection{Proof of Lemma \ref{lma:2.4}}\label{pflma:2.4}
  As $\pi\in\Pi$, there exists $p>1$ such that $\pi\in\Pi_p$. Using Hölder's inequality,  we have
    $$
    \mathbb{E}_t\left[\exp\left(-\zeta(R^{\pi}_T-R^{\pi}_t)\right)\right]\le\left(\mathbb{E}_t\left[\exp\left(-p\zeta(R^{\pi}_T-R^{\pi}_t)\right)\right]\right)^{\frac{1}{p}}< \infty \quad a.s..
    $$    
  By the SDE \eqref{eq:R} of $R^{\pi}$, it holds that
    $$
    R^{\pi}_T-R^{\pi}_t= \int_t^T \left[r_s+\pi_s(\mu_s-r_s)-\frac{1}{2}|\sigma_s \pi_s|^2 \right]\md s+\pi_s\sigma_s\md B_s.
     $$
    Using  standard estimate for SDE (see Theorem 3.4.3 in \cite{zhang_backward_2017}), for any $A\in\mathcal{F}_t$ we obtain
    $$
 \begin{aligned}
    \mathbb{E}\left[\sup\limits_{s\in [t,T]}\left|R^{\pi}_s-R^{\pi}_t\right|^2\mathds{1}_{A}\right]\leq& C\mathbb{E}\left[\left(\int_t^T \left[r_s+\pi_s(\mu_s-r_s)-\frac{1}{2}|\sigma_s \pi_s|^2 \right]\mathds{1}_{A}\md s\right)^2+\int_t^T|\pi_s\sigma_s|^2\mathds{1}_{A}\md s\right]\\
    \le& C\mathbb{E}\left[\left(\left(\int_t^T|\pi_s|^2\md s\right)^2+1\right)\mathds{1}_{A}\right].
    \end{aligned}
 $$
  Because $\pi\in H_{\bmo}$, by the definition of conditional expectation and Lemma \ref{lma:bmo:property}, we have
 \begin{align*}
 &\left(\mathbb{E}_t\left[\left|R^{\pi}_T-R^{\pi}_t\right|\right]\right)^2\le\mathbb{E}_t\left[\left|R^{\pi}_T-R^{\pi}_t\right|^2\right]\\ \le&\mathbb{E}_t\left[\sup\limits_{s\in [t,T]}\left|R^{\pi}_s-R^{\pi}_t\right|^2\right] \le C\left(\mathbb{E}_t\left[\left(\int_t^T|\pi_s|^2\md s\right)^2\right]+1\right) \\
 \le & C\left(\|\pi\|_{\textup{BMO}}^4+ 1\right)<\infty.
 \end{align*}
 Hence $\textup{Var}_t[R^{\pi}_T-R^{\pi}_t]$ is  bounded, and consequently $J(t,\pi)$ is finite  for any $t\in[0,T)$. The argument also verifies that $R^{\pi}\in L^2_{\mathbb{F}}(\Omega; C([0, T]; \mathbb{R}))$.

  \subsection{Proof of Lemma \ref{diffpro}}\label{pflm:diffpro}
  
      The first assertion is standard; see Theorem 3.4.3 in \cite{zhang_backward_2017}.

For the second assertion, the argument is similar to the proof of Lemma \ref{lma:2.4}. For any set $A \in \mathcal{F}_t$, we have
 $$
 \begin{aligned}
\mathbb{E}\left[\sup\limits_{s\in[t,T]}|\xi_s^{t,\varepsilon,\pi}|^{2k}\mathds{1}_{A}\right]\le& C\mathbb{E}\left[\left(\int_t^{t+\varepsilon}\left|(\mu_s-r_s)\eta-\frac{1}{2}\sigma_s^2\eta^2-\sigma_s^2\pi_s\eta\right|\mathds{1}_{A}\md s\right)^{2k}+\left(\int_t^{t+\varepsilon}|\sigma_s\eta|^2 \mathds{1}_A\md s\right)^k \right]\\
\le&C\mathbb{E}\left[\left\{\left(\varepsilon+\varepsilon\eta+\varepsilon^{\frac{1}{2}}\left(\int_t^{t+\varepsilon}|\pi|^2\md s\right)^{\frac{1}{2}}\right)^{2k}|\eta|^{2k}+\varepsilon^k|\eta|^{2k}\right\}\mathds{1}_A\right]\\
=& C\mathbb{E}\left[\left((\varepsilon|\eta|^2)^k\left[\left(\int_t^{t+\varepsilon}|\pi_s|^2\md s\right)^{k}+1\right]+o(\varepsilon^k)\right)\mathds{1}_A\right].
\end{aligned}
$$
 Then it follows that
$$
\mathbb{E}_t\left[\sup\limits_{s\in[t,T]}|\xi_s^{t,\varepsilon,\pi}|^{2k}\right]\le C(\varepsilon|\eta|^2)^k\left(1+\mathbb{E}_t\left[\left(\int_t^T|\pi_s|^2\md s\right)^{k}\right]\right)\le C(k,\|\pi\|_{\textup{BMO}})(\varepsilon|\eta|^2)^k,
$$
where we have used the fact that $\pi\in H_{\textup{BMO}}$ and the conditional version of the energy inequality (see Lemma \ref{lma:bmo:property}). 

For the third assertion, note that $\mathbb{P}-a.s.$
$$
\begin{aligned}
&\mathbb{E}_t\left[\exp(c\xi^{t,\varepsilon,\pi}_T)\right]\\
=& \mathbb{E}_t\left[\exp\left(c\int_t^{t+\varepsilon}\left((\mu_s-r_s)\eta-\frac{1}{2}\sigma_s^2\eta^2-\sigma_s^2\pi_s\eta\right)\md s+c\int_t^{t+\varepsilon}\sigma_s\eta\md B_s\right)\right] \\
\le& \left(\mathbb{E}_t\left[\exp\left(2c\int_t^{t+\varepsilon}\left((\mu_s-r_s)\eta-\frac{1}{2}|\sigma_s\eta|^2-\sigma_s^2\pi_s\eta\right)\md s+2c^2\int_t^{t+\varepsilon}|\sigma_s\eta|^2\md s\right)\right]\right)^{\frac{1}{2}}\\
&\cdot\left(\mathbb{E}_t\left[\exp\left(2c\int_t^{t+\varepsilon}\sigma_s\eta\md B_s-2c^2\int_t^{t+\varepsilon}|\sigma_s\eta|^2\md s\right)\right]\right)^{\frac{1}{2}}\\
\le& C(|c|,\|\eta\|_{\infty}) \left(\mathbb{E}_t\left[\exp\left(2|c|\int_t^T\|\sigma\|^2_{\infty}\|\eta\|_{\infty}|\pi_s|\md s\right)\right]\right)^{\frac{1}{2}},
\end{aligned}
$$
where the first inequality follows from the Cauchy–Schwarz inequality.
For the second inequality, since $\sigma$ and $\eta$ are bounded, Novikov’s condition is satisfied, and hence
\begin{align*}  \mathbb{E}_t\left[\exp\left(2c\int_t^{t+\varepsilon}\sigma_s\eta\md B_s-2c^2\int_t^{t+\varepsilon}|\sigma_s\eta|^2\md s\right)\right]=1\quad a.s..
\end{align*}
Moreover, we can choose
$$
C(|c|,\|\eta\|_{\infty}):= \exp\left(|c|(T-t)\left(\|\mu-r\|_{\infty}\cdot\|\eta\|_{\infty}+\frac{1}{2}\|\sigma\|_{\infty}^2\|\eta\|_{\infty}^2\right)+c^2(T-t)\|\sigma\|_{\infty}^2\|\eta\|_{\infty}^2\right).
$$
Then, by another application of the Cauchy–Schwarz inequality,
$$
\begin{aligned}
   \mathbb{E}_t\left[\exp (c\xi^{t,\varepsilon,\pi}_T)\right] \le& C(|c|,\|\eta\|_{\infty})\left(\mathbb{E}_t\left[\exp\left(\int_t^T\left(\frac{\left(c\|\sigma\|^2_{\infty}\|\eta\|_{\infty}\right)^2}{\delta}+\delta|\pi_s|^2\right)\md s\right)\right]\right)^{\frac{1}{2}}\\
   \le&C(|c|,\|\eta\|_{\infty})\exp\left(\frac{T-t}{2\delta}\left(c\|\sigma\|^2_{\infty}\|\eta\|_{\infty}\right)^2\right)\left(\mathbb{E}_t\left[\exp\left(\delta\int_t^T|\pi_s|^2\md s\right)\right]\right)^{\frac{1}{2}} \quad a.s..
\end{aligned}
$$
By the John–Nirenberg inequality (see Lemma \ref{John-Nirenberg}), for $\delta>0$ such that $\delta\|\pi\|_{\textup{BMO}}^2<1$, we have
$$
\mathbb{E}_t\left[\exp\left(\delta\int_t^T|\pi_s|^2\md s\right)\right]<\frac{1}{1-\delta\|\pi\|_{\textup{BMO}}^2}<\infty\quad a.s..
$$
Taking $\delta=\frac{1}{2\|\pi\|_{\textup{BMO}}^2}$ yields
\begin{align*}
\sup\limits_{\varepsilon\in (0,T-t)}\mathbb{E}_t\left[\exp(c\xi^{t,\varepsilon,\pi}_T)\right]
\leq C(|c|,\|\eta\|_{\infty},\|\pi\|_{\textup{BMO}})\quad a.s.,
\end{align*}
and the constant $C$ is non-decreasing in $\|\eta\|_{\infty}$.

\subsection{Proof of Corollary \ref{coro:admi:piepsilon}}\label{pfcoro:admi:piepsilon}
  As $\pi\in\Pi$, there exists $p>1$ such that $\pi\in\Pi_p$. We can choose $1<p^{\prime}<p$ and $q>1$ such that $\frac{1}{p}+\frac{1}{q}=\frac{1}{p^{\prime}}$.  By Lemma \ref{diffpro}, it holds that

     $$
     \begin{aligned}
      &\left(\mathbb{E}_t\left[\exp\left(-p^{\prime}\zeta(R^{\pi^{t,\varepsilon,\eta}}_T-R^{\pi^{t,\varepsilon,\eta}}_t)\right)\right]\right)^{\frac{1}{p^{\prime}}} \\=&\left(\mathbb{E}_t\left[\exp\left(-p^{\prime}\zeta\left[(R^{\pi}_T-R^{\pi}_t)+\xi_T^{t,\varepsilon,\pi}\right]\right)\right]\right)^{\frac{1}{p^{\prime}}}\\
      \le&\left(\mathbb{E}_t\left[\exp\left(-p\zeta(R^{\pi}_T-R^{\pi}_t)\right)\right]\right)^{\frac{1}{p}}\left(\mathbb{E}_t\left[\exp\left(-q\zeta\xi_T^{t,\varepsilon,\pi}\right)\right]\right)^{\frac{1}{q}}<\infty \;a.s..
      \end{aligned}
     $$
 Moreover, it is obvious that under a bounded perturbation $\eta$, 
 $$
 \|\pi^{t,\varepsilon,\eta}\|_{\textup{BMO}}^2\le2\sup_{\tau \in \mathcal{T}_{[0,T]}} \left\| \mathbb{E}_{\tau} \left[ \int_\tau^T |\pi_s|^2 ds \right] \right\|_\infty + 2\varepsilon\|\eta\|_{\infty}^2<\infty.
 $$
 Thus, $\pi^{t,\varepsilon,\eta}\in\Pi_{p'}\subset\Pi$.

\subsection{Proof of Lemma \ref{lma:3.4}}\label{pflma:3.4}
Fix  $\pi\in\Pi$, $t\in [0,T)$, $\eta\in L_{\mathcal{F}_t}^{\infty}(\Omega,\mathbb{R}^d)$ and $\varepsilon\in (0,T-t)$. By the  definition of $J$, we have
$$
    \begin{aligned}
         &J(t,\pi^{t,\varepsilon})-J(t,\pi)
         =\mathbb{E}_t\left[U(R_T^{t,\varepsilon}-R_t)-U(R_T-R_t)\right]-\frac{\gamma}{2}\left(\textup{Var}_t[R_T^{t,\varepsilon}]-\textup{Var}_t[R_T]\right).
    \end{aligned}
    $$
    For the first term, we have that
\begin{align}\label{eq:3.4:step1
}\mathbb{E}_t\left[U(R_T^{t,\varepsilon}-R_t)-U(R_T-R_t)\right]
    =e^{\zeta R_t}\mathbb{E}_t\left[U^{'}(R_T)\xi_T^{t,\varepsilon}+\int_0^1 U^{''}(R_T+\lambda\xi_T^{t,\varepsilon})(1-\lambda)d\lambda|\xi_T^{t,\varepsilon}|^2\right].
\end{align}
Next we analyze the conditional expectation $\mathbb{E}_t\left[U^{'}(R_T)\xi_T^{t,\varepsilon}\right]$. To this end, let $(\alpha,\beta)$ be the unique adapted solution of the BSDE
$$
\left\{
\begin{aligned}
    &\md \alpha_s=\beta_s^1\md B_s+\beta_s\md \bar{B}_s ,\\
    &\alpha_T=e^{-\zeta R_T}.
\end{aligned}
\right.
$$
The martingale representation theorem impllies that $\alpha_s=\mathbb{E}_s\left[U^{'}(R_T)\right]$. By the definition of the admissible strategy, it holds that $\alpha_T\in L^p_{\mathcal{F}_T}(\Omega;\mR)$ for some $p>1$. Let $Y=-\frac{1}{\zeta}\log\alpha-R$ and hence $U^{'}(R_s+Y_s)=\alpha_s=\mathbb{E}_s\left[U^{'}(R_T)\right]$. In particular, $\alpha$ is positive and $Y$ is well defined for any $t\in [0,T]$ a.s.. Moreover, $\alpha\in L^p_{\mathbb{F}}(\Omega;C([0,T];\mathbb{R}))$ and $\beta^1, \beta\in L^p_{\mathbb{F}}\left(\Omega;L^2(0,T;\mathbb{R})\right)$. Then $Y$ is an adapted process satisfying the following BSDE that
$$
\left\{
\begin{aligned}
    &\md Y_s=\left(\frac{1}{2\zeta \alpha_s^2}\left(|\beta_s^1|^2+|\beta_s|^2+2\rho\beta_s^1\beta_s\right)-a(s,\pi_s)\right)\md s-\left(\frac{\beta_s^1}{\zeta\alpha_s}+\sigma\pi_s\right)\md B_s-\frac{\beta_s}{\zeta\alpha_s}\md \bar{B}_s ,\\
    &Y_T=0.
\end{aligned}
\right.
$$
Let $Z^1=-\frac{\beta^1}{\zeta\alpha}-\sigma\pi$ and $Z=-\frac{\beta}{\zeta\alpha}$, it then holds that 
$$
\md Y_s=\left(\frac{\zeta}{2}|Z_s^1+\sigma_s\pi+\rho Z_s|^2+\frac{\zeta(1-\rho^2)}{2}|Z_s|^2-a(s,\pi_s)\right)\md s+Z_s^1\md B_s+Z^1\md \bar{B}_s.
$$
Moreover,  taking the conditional expectations, we see that 
$$
\mathbb{E}_t\left[U^{'}(R_T)\xi_T^{t,\varepsilon}\right]=\mathbb{E}_t[\alpha_{T}\xi_{t+\varepsilon}^{t,\varepsilon}]=\mathbb{E}_t\left[\mathbb{E}[\alpha_{T}|\mathcal{F}_{t+\varepsilon}]\xi_{t+\varepsilon}^{t,\varepsilon}\right]=\mathbb{E}_t[\alpha_{t+\varepsilon}\xi_{t+\varepsilon}^{t,\varepsilon}]=\mathbb{E}_t\left[U^{'}(R_{t+\varepsilon}+Y_{t+\varepsilon})\xi_{t+\varepsilon}^{t,\varepsilon}\right].
$$
It$\hat{\text{o}}$'s formula gives that
\begin{align}\label{cp1}
    &\mathbb{E}_t\left[U^{'}(R_{t+\varepsilon}+Y_{t+\varepsilon})\xi_{t+\varepsilon}^{t,\varepsilon}\right]\nonumber\\
    =&\mathbb{E}_t\left[\int_t^{t+\varepsilon}U^{'}(R_s+Y_s)\left(\left(a(s,\pi^{t,\varepsilon}_s)-a(s,\pi_s)\right)\md s+\sigma_s\eta \md B_s\right)\right.\nonumber\\
    &+\int_t^{t+\varepsilon}\xi_s^{t,\varepsilon}U^{''}(R_s+Y_s)\left((Z_s^1+\sigma_s\pi_s)\md B_s+Z_s\md \bar{B}_s+\left(\frac{\zeta}{2}|Z_s^1+\sigma_s\pi_s+\rho Z_s|^2+\frac{\zeta(1-\rho^2)}{2}|Z_s|^2\right)\md s\right)\nonumber\\
    &+\frac{1}{2}\int_t^{t+\varepsilon}\xi_s^{t,\varepsilon}U^{'''}(R_s+Y_s)\left(|Z_s^1+\sigma_s\pi_s|^2+|Z_s|^2+2\rho(Z_s^1+\sigma_s\pi_s)Z_s\right)\md s\nonumber\\
    &\left.+\int_t^{t+\varepsilon}U^{''}(R_s+Y_s)(Z_s^1+\sigma_s\pi_s+\rho Z_s)\cdot \sigma_s\eta \md s\right]\nonumber\\
    =&\mathbb{E}_t\left[\int_t^{t+\varepsilon}e^{-\zeta(R_s+Y_s)}\left(a(s,\pi^{t,\varepsilon}_s)-a(s,\pi_s)-\zeta(Z_s^1+\sigma_s\pi_s+\rho Z_s)\cdot \sigma_s\eta \right)\md s\right]\nonumber\\
    &+\mathbb{E}_t\left[\int_t^{t+\varepsilon}e^{-\zeta(R_s+Y_s)}\left[(\sigma_s\eta-\zeta\xi_s^{t,\varepsilon}(Z_s^1+\sigma_s\pi_s))\md B_s-\zeta\xi_s^{t,\varepsilon}Z_s\md \bar{B}_s\right]\right].
\end{align}
Under Assumption \ref{assumpmarket}, we have 
$$\mathbb{E}_t\left[\left(\int_t^{t+\varepsilon}(\alpha_s\sigma_s\eta)^2\md s\right)^{\frac{1}{2}}\right]\le \left(\mathbb{E}_t\left[\left(\sup\limits_{s\in[t,t+\varepsilon]}\alpha_s\right)^p\right]\right)^{\frac{1}{p}}\left(\mathbb{E}_t\left[\left(\sup\limits_{s\in [t,t+\varepsilon]}\sigma_s\eta\right)^q\right]\right)^{\frac{1}{q}}<\infty,$$
$$\mathbb{E}_t\left[\left(\int_t^{t+\varepsilon}(\beta_s\xi_s^{t,\varepsilon})^2\md s\right)^{\frac{1}{2}}\right]\le \left(\mathbb{E}_t\left[\left(\int_t^{t+\varepsilon}(\beta_s)^2\md s\right)^{\frac{p}{2}}\right]\right)^{\frac{1}{p}}\left(\mathbb{E}_t\left[\left(\sup\limits_{s\in [t,t+\varepsilon]}\xi_s^{t,\varepsilon}\right)^q\right]\right)^{\frac{1}{q}}<\infty,$$
$$\mathbb{E}_t\left[\left(\int_t^{t+\varepsilon}(\beta_s^1\xi_s^{t,\varepsilon})^2\md s\right)^{\frac{1}{2}}\right]\le \left(\mathbb{E}_t\left[\left(\int_t^{t+\varepsilon}(\beta_s^1)^2\md s\right)^{\frac{p}{2}}\right]\right)^{\frac{1}{p}}\left(\mathbb{E}_t\left[\left(\sup\limits_{s\in [t,t+\varepsilon]}\xi_s^{t,\varepsilon}\right)^q\right]\right)^{\frac{1}{q}}<\infty.$$
The Burkholder-Davis-Gundy inequality yields that $$\mathbb{E}_t\left[\int_t^{t+\varepsilon}e^{-\zeta(R_s+Y_s)}\left[(\sigma\eta-\zeta\xi_s^{t,\varepsilon}(Z_s^1+\sigma\pi))\md B_s-\zeta\xi_s^{t,\varepsilon}Z_s\md \bar{B}_s\right]\right]=0.$$
For the second term,  let $\tilde{\alpha}_t=\mathbb{E}_t[R_T]$ and $\tilde{Y_t}=\tilde{\alpha}_t-R_t$. Then, $\tilde{\alpha}$ and $\tilde{Y}$ satisfy that
\begin{align*}
&\md\tilde{\alpha}=\tilde{\beta}^1\md B_s+\tilde{\beta}\md \bar{B}_s,\\
&\md \tilde{Y}_s=-a(s,\pi_s)\md s+\tilde{Z}^1_s\md B_s+\tilde{Z}_s\md \bar{B}_s,
\end{align*}
and $\tilde{\alpha}\in L^2_{\mathbb{F}}(\Omega; C([0, T]; \mathbb{R})), \tilde{\beta}^1=\tilde{Z}^1+\sigma\pi, \tilde{\beta}=\tilde{Z}\in L^2_{\mathbb{F}}\left(\Omega;L^2(0,T;\mathbb{R})\right)$. From Lemma \ref{lma:2.4}, it follows that $R^{\pi}\in L^2_{\mathbb{F}}(\Omega; C([0, T]; \mathbb{R}))$ and hence $\tilde{Y}\in L^2_{\mathbb{F}}(\Omega; C([0, T]; \mathbb{R}))$. We thus have
\begin{align}\label{cp2}
 &\textup{Var}_t[R_T]-\textup{Var}_t[R_T^{t,\varepsilon}]\nonumber\\
 =&\textup{Var}_t\left[\mathbb{E}_{t+\varepsilon}[R_T]\right]+\mathbb{E}_t\left[\textup{Var}_{t+\varepsilon}[R_T]\right]-\textup{Var}_t\left[\mathbb{E}_{t+\varepsilon}[R_T^{t,\varepsilon}]\right]-\mathbb{E}_t\left[\textup{Var}_{t+\varepsilon}[R_T^{t,\varepsilon}]\right]\nonumber\\
=&\textup{Var}_t[R_{t+\varepsilon}+\tilde{Y}_{t+\varepsilon}]-\textup{Var}_t[R_{t+\varepsilon}^{t,\varepsilon}+\tilde{Y}_{t+\varepsilon}]\nonumber\\
   =&\textup{Var}_t\left[\int_t^{t+\varepsilon}(\sigma\pi_s+\tilde{Z}^1_s)\md B_s+\tilde{Z}_s\md \bar{B}_s\right]-\textup{Var}_t\left[\int_t^{t+\varepsilon}(\sigma\pi_s^{t,\varepsilon}+\tilde{Z}_s^1)\md B_s+\tilde{Z}_s\md \bar{B}_s\right]\nonumber\\
   =&\mathbb{E}_t\left[\int_t^{t+\varepsilon}((\sigma\pi_s+\tilde{Z}^1_s+\rho\tilde{Z}_s)^2-(\sigma\pi_s^{t,\varepsilon}+\tilde{Z}^1_s+\rho\tilde{Z}_s)^2)\md s\right],
\end{align}
where we have used the fact that $R_T-R_{t+\varepsilon}=R_T^{t,\varepsilon}-R_{t+\varepsilon}^{t,\varepsilon}$ in the second and third equalities . The conclusion then follows by combining \eqref{eq:3.4:step1
}, \eqref{cp1} and \eqref{cp2}.
\subsection{Proof of Corollary \ref{corol:verifyequilibrium}}\label{pfcoro:verifyequilibrium}
In contrast to the proof of Lemma \ref{lma:3.4}, where the processes $\alpha_s:=\mathbb{E}_s\left[U^{'}(R_T)\right]$ and $\tilde{\alpha}_t:=\mathbb{E}_t[R_T]$ were defined first and subsequently used to construct the processes $Y$ and $\tilde{Y}$, here we start from the solution $(Y,\tilde{Y})$ to the simplified BSDE \eqref{eq:simplifiedBSDE} and define $\alpha_s:=U^{'}(R_s+Y_s)$ and $\tilde{\alpha}_s:=\tilde{Y}_s+R_s$. Applying Itô’s formula yields
$$
\begin{aligned}
    \md \alpha_s=&-\zeta\alpha_s\md(R_s+Y_s)+\frac{1}{2}\zeta^2\alpha_s\langle R+Y\rangle_s\\
    =&-\zeta\alpha_s\left(\frac{\zeta}{2}|\sigma_s\pi_s+\rho Z_s|^2+\frac{\zeta(1-\rho^2)}{2}|Z_s|^2\right)\md s-\zeta\alpha_s\left(\sigma_s\pi_s\md B_s+Z_s\md \bar{B}_s\right)\\
    &+\frac{1}{2}\zeta^2\alpha_s\left(|\sigma_s\pi_s|^2+|Z_s|^2+2\rho\sigma_s\pi_sZ_s\right)\md s\\
    =& -\zeta\alpha_s\left(\sigma_s\pi_s\md B_s+Z_s\md \bar{B}_s\right).
\end{aligned}
$$
and similarly, 
$$
\begin{aligned}
    \md \tilde{\alpha}_s=\md \tilde{Y}_s+\md R_s
    =\sigma_s\pi_s\md B_s+\tilde{Z}_s\md \bar{B}_s.
\end{aligned}
$$
Hence, by the martingale representation and terminal conditions of $Y$ and $\tilde{Y}$, we recover $\alpha_s=\mathbb{E}_s\left[U^{'}(R_T)\right]$ and $\tilde{\alpha}_t=\mathbb{E}_t[R_T]$. Because these processes retain the same structural properties as their counterparts in the proof of Lemma \ref{lma:3.4}, the remaining arguments require only minor adjustments, and the desired result follows.

\subsection{Proof of Lemma \ref{lma:3.5}}\label{pflma:3.5}
   Hölder's inequality and Lemma \ref{diffpro} imply that  $\mathbb{E}_t\left[e^{-\zeta R_T}|\xi_T^{t,\varepsilon}|^2\right]<\infty$ and $ \mathbb{E}_t\left[e^{-\zeta(R_T+\xi_T^{t,\varepsilon})}|\xi_T^{t,\varepsilon}|^2\right]<\infty$ a.s.. Noting that $e^{-(R_T+\lambda\xi_T^{t,\varepsilon})}|\xi_T^{t,\varepsilon}|^2$ is nonnegative and convex almost sure with respect to $\lambda$, we have
     $$
     \begin{aligned}
         &\frac{1}{\varepsilon} \mathbb{E}_t\left[\int_0^1 \left(e^{-\zeta(R_T+\lambda\xi_T^{t,\varepsilon}-R_t)}(1-\lambda)\right)d\lambda\cdot|\xi_T^{t,\varepsilon}|^2\right] 
         =\frac{1}{\varepsilon}\int_0^1\left(\mathbb{E}_t\left[ e^{-\zeta(R_T+\lambda\xi_T^{t,\varepsilon}-R_t)}|\xi_T^{t,\varepsilon}|^2\right](1-\lambda)\right)\md\lambda \\
         \le & \frac{1}{\varepsilon}e^{\zeta R_t}\left(\mathbb{E}_t\left[e^{-\zeta R_T}|\xi_T^{t,\varepsilon}|^2\right]\int_0^1(1-\lambda)^2\md \lambda+\mathbb{E}_t\left[e^{-\zeta (R_T+\xi_T^{t,\varepsilon})}|\xi_T^{t,\varepsilon}|^2\right]\int_0^1\lambda(1-\lambda)\md\lambda\right).
     \end{aligned}
     $$
     Therefore, we only need to show that $\frac{1}{\varepsilon} \mathbb{E}_t\left[e^{-(R_T+\lambda\xi_T^{t,\varepsilon}-R_t)}|\xi_T^{t,\varepsilon}|^2\right]\le C|\eta|^2 $ a.s. for $\lambda=0,1$, and for any  $t\in[0,T)$ and $\eta\in L_{\mathcal{F}_t}^{\infty}(\Omega,\mathbb{R})$. By Lemma \ref{diffpro} and the Hölder's inequality, there exists a constant $C(\|\eta\|_{\infty})$  such that,
for $\varepsilon$ sufficiently small,
$$
\begin{aligned}
\frac{1}{\varepsilon} \mathbb{E}_t\left[e^{-\zeta(R_T+\lambda\xi_T^{t,\varepsilon}-R_t)}|\xi_T^{t,\varepsilon}|^2\right]&\le \frac{1}{\varepsilon} \left(\mathbb{E}_t\left[e^{-p\zeta(R_T-R_t)}|\right]\right)^{\frac{1}{p}} \left(\mathbb{E}_t\left[|\xi_T^{t,\varepsilon}|^{2q}\right]\right)^{\frac{1}{q}}\left(\mathbb{E}_t\left[e^{-r\zeta\lambda\xi_T^{t,\varepsilon}}\right]\right)^{\frac{1}{r}}\\
&\le C\left(\mathbb{E}_t\left[e^{-p\zeta(R_T-R_t)}|\right]\right)^{\frac{1}{p}}\left(\mathbb{E}_t\left[e^{-r\zeta\lambda\xi_T^{t,\varepsilon}}\right]\right)^{\frac{1}{r}} |\eta|^2\\
&\le C(\|\eta\|_{\infty})|\eta|^2, \quad \lambda=0, 1.
\end{aligned}
$$
Here, we choose $q,r$ such that $\frac{1}{p}+\frac{1}{q}+\frac{1}{r}=1$.
Then, it holds that
$$
 \limsup\limits_{\varepsilon\rightarrow 0}\frac{1}{\varepsilon} \mathbb{E}_t\left[\int_0^1 e^{-\zeta (R_T+\lambda\xi_T^{t,\varepsilon}-R_t)}(1-\lambda)d\lambda|\xi_T^{t,\varepsilon}|^2\right]\le C(\|\eta\|_{\infty})|\eta|^2 \quad a.s..
$$
\subsection{Proof of Lemma \ref{lma:3.6}}\label{pflma:3.6}
Noting that $a(s,\pi^{t,\varepsilon_n^t}_s)-a(s,\pi_s)=(\mu_s-r_s)\eta-\frac{1}{2}\sigma^2_s\eta^2-\sigma^2_t\eta\pi_s$ and $(\sigma_s\pi_s+\tilde{Z}^1_s+\rho\tilde{Z}_s)^2-(\sigma_s\pi_s^{t,\varepsilon}+\tilde{Z}^1_s+\rho\tilde{Z}_s)^2=-2\sigma_s\eta(\sigma_s\pi_s+\tilde{Z}^1_s+\rho\tilde{Z}_s-\sigma_s^2\eta^2)$, we only need to prove that
$$
\begin{aligned}
 &\lim\limits_{n\rightarrow\infty} \frac{1}{\varepsilon_n^t}\mathbb{E}_t\left[\int_t^{t+\varepsilon_n^t} \alpha_s\left((\mu_s-r_s)\eta-\frac{1}{2}\sigma^2_s\eta^2-\sigma^2_t\eta\pi_s-\zeta(Z_s^1+\sigma_s\pi_s+\rho Z_s)\cdot \sigma_s\eta \right)\md s\right]\\
 =&\alpha_t\left((\mu_t-r_t)\eta-\frac{1}{2}\sigma^2_t\eta^2-\sigma^2_t\eta\pi_t-\zeta(Z_t^1+\sigma_t\pi_t+\rho Z_t)\cdot \sigma_t\eta\right)\quad a.s.
 \end{aligned}
$$
and 
$$
\lim\limits_{n\rightarrow\infty} \frac{1}{\varepsilon_n^t}\mathbb{E}_t\left[\int_t^{t+\varepsilon_n^t} 2\sigma_s(\sigma_s\pi_s+\tilde{Z}^1_s+\rho\tilde{Z}_s-\sigma_s^2\eta^2)\md s\right]=2\sigma_t(\sigma_t\pi_t+\tilde{Z}^1_t+\rho\tilde{Z}_t-\sigma_t^2\eta^2)\quad a.s..
$$
In view of $ e^{-\zeta(R_s+Y_s)}=\alpha_s \in L^p_{\mathbb{F}}(\Omega;C([0,T];\mathbb{R}))$ and the fact that $\mu-r$ and $\sigma$ are bounded, we have $\alpha(\mu-r),\alpha\sigma^2\in L^{p}(0,T; \mathbb{R})$ with $p>1$. It also holds that, for any $1<\gamma<p$ and $\frac{1}{\gamma}=\frac{1}{p}+\frac{1}{q^{\prime}}$,
       $$
       \mathbb{E}\left[\int_0^T |\alpha_s\pi_s|^{\gamma}\md s\right]\le \left(\mathbb{E}\left[\sup\limits_{s\in[0,T]}|\alpha_s|^p\right]\right)^{\frac{\gamma}{p}}\left(\mathbb{E}\left[\left(\int_0^T |\pi_s|^{\gamma}\md s\right)^{\frac{q^{\prime}}{\gamma}}\right]\right)^{\frac{\gamma}{q^{\prime}}}<\infty,
       $$
       which implies that $\alpha\sigma^2\pi\in L^{\gamma}(0,T; \mathbb{R})$. In addition,  $\alpha(Z^1+\sigma\pi)=\beta^1,\; \alpha Z=\beta\in L^{ p}_{\mathbb{F}}\left(\Omega;L^2(0,T;\mathbb{R})\right)\subset L^{p\wedge 2}(0,T; \mathbb{R})$ and $\tilde{\beta}^1=\tilde{Z}^1+\sigma\pi, \tilde{\beta}=\tilde{Z}\in L^2_{\mathbb{F}}\left(\Omega;L^2(0,T;\mathbb{R})\right)\subset L^{2}(0,T; \mathbb{R})$.

       Therefore,  by Lemma 3.3 in \cite{Hamaguchi_2021} , there exists a measurable set $E_1$ with $\operatorname{Leb}([0,T] \setminus E_1) = 0$ such that for any $t\in E_1$, we can choose a subsequence by recursively extracting subsequences and applying the diagonal selection argument that
       $$
\lim\limits_{n\rightarrow\infty}\mathbb{E}_t \left[
\frac{1}{\varepsilon_n^t}\int_t^{\varepsilon_n^t}P_s\right]=P_t.
       $$
       Here $P$ can be chosen to be $\alpha(\mu-r), \alpha\sigma^2, \alpha\sigma^2\pi, \alpha(Z^1+\sigma\pi+\rho Z)$ and $\sigma(\sigma\pi+\tilde{Z}^1+\rho\tilde{Z}), \sigma^3\eta^2$. Then we obtain \eqref{lma:3.5.1} and \eqref{lma:3.5.2} with $E_1$ and the corresponding $\varepsilon_n^t$ for any $t\in E_1$, independent of $\eta$.

\subsection{Proof of Lemma \ref{lma:projection}}\label{pflma:projection}
  We first prove assertion (i).  By the characterization of the orthogonal projection onto a convex set,
for every $z\in U$, we have
\[
\langle w-u,\,z-u\rangle \le 0.
\]
Choosing $z=u+h\in U$ yields
\begin{equation}\label{inner_nonpos}
\langle w-u,\,h\rangle \le 0.
\end{equation}
It then follows that
\[
\begin{aligned}
\bigl|\alpha(w-u)-h\bigr|^2
= \alpha^2|w-u|^2 + |h|^2 - 2\alpha\langle w-u,\,h\rangle 
\ge \alpha^2|w-u|^2.
\end{aligned}
\]
Taking the square root yields the desired inequality \eqref{eq:alpha:convex}. Moreover, the equality requires both $\langle w-u,h\rangle=0$ and $|h|^2=0$, which implies that
$h=0$.

Next we verify assertion (ii). As $u$ is not the projection, the vector $v=P_U(w)$ satisfies that
\[
\langle w-u,\,v-u\rangle \ge |v-u|^2 > 0.
\]
Let $h=\lambda(v-u),\;\lambda\in(0,1]$, then
$$
\bigl|\alpha(w-u)-h\bigr|<\alpha|w-u| \iff 2\alpha\langle w-u, h\rangle >  |h|^2 \Longleftarrow 0<\lambda<2\alpha
$$
because $\alpha>0$, the above inequality holds for sufficiently small $\lambda>0$. By the convexity of $U$, $u+ h\in U$, which completes the proof.
 \end{document}